\documentclass [14pt]{amsart}

\usepackage{amsmath}
\usepackage{amssymb}
\usepackage{amscd}
\usepackage{epsfig}
\usepackage{amsxtra}
\usepackage{pifont}
\usepackage{mathabx}
\usepackage{MnSymbol}
\usepackage{accents}
\usepackage{scalerel,stackengine}
\stackMath
\newcommand\reallywidehat[1]{%
\savestack{\tmpbox}{\stretchto{%
  \scaleto{%
    \scalerel*[\widthof{\ensuremath{#1}}]{\kern-.6pt\bigwedge\kern-.6pt}%
    {\rule[-\textheight/2]{1ex}{\textheight}}
  }{\textheight}%
}{0.5ex}}%
\stackon[1pt]{#1}{\tmpbox}%
}

\makeatletter
\providecommand*{\cupdot}{%
  \mathbin{%
    \mathpalette\@cupdot{}%
  }%
}
\newcommand*{\@cupdot}[2]{%
  \ooalign{%
    $\m@th#1\cup$\cr
    \hidewidth$\m@th#1\cdot$\hidewidth
  }%
}
\makeatother

\makeatletter
\providecommand*{\bigcupdot}{%
  \mathbin{%
    \mathpalette\@bigcupdot{}%
  }%
}
\newcommand*{\@bigcupdot}[2]{%
  \ooalign{%
    $\m@th#1\bigcup$\cr
    \hidewidth$\m@th#1\cdot$\hidewidth
  }%
}
\makeatother

\theoremstyle{bold}
\newtheorem{theorem}{Theorem}[section]
\newtheorem{proposition}[theorem]{Proposition}
\newtheorem{lemma}[theorem]{Lemma}
\newtheorem{corollary}[theorem]{Corollary}

\newtheorem{definition}[theorem]{Definition}
\newtheorem{remark}[theorem]{Remark}
\newtheorem{example}[theorem]{Example}

\newcommand{\supp}{\mathrm{supp}}

\DeclareMathOperator{\dens}{dens}

\newcommand{\cD}{\mathcal{D}}

\newcommand{\cL}{\mathcal{L}}
\newcommand{\cM}{\mathcal{M}}

\newcommand{\ZZ}{\mathbb{Z}}

\newcommand{\RR}{\mathbb{R}}
\newcommand{\CC}{\mathbb{C}}

\newcommand{\oplam}{\mbox{\Large $\curlywedge$}}

\begin{document}

\title[Pure point diffraction in cut-and-project sets]{A  short guide to pure point diffraction \\ in cut-and-project sets}

\dedicatory{We dedicate this work to Tony Guttmann on the occasion of his $70^{th}$ birthday.}

\author{Christoph Richard}
\address{Department f\"{u}r Mathematik, Friedrich-Alexander-Universit\"{a}t Erlangen-N\"{u}rnberg,
Cauerstrasse 11, 91058 Erlangen, Germany}
\email{christoph.richard@fau.de}

\author{Nicolae Strungaru}
\address{Department of Mathematical Sciences, MacEwan University \\
10700 “ 104 Avenue, Edmonton, AB, T5J 4S2\\
and \\
Department of Mathematics\\
Trent University \\
Peterborough, ON
and \\
Institute of Mathematics ``Simon Stoilow''\\
Bucharest, Romania}
\email{strungarun@macewan.ca}
\urladdr{http://academic.macewan.ca/strungarun/}

\renewcommand{\thefootnote}{\fnsymbol{footnote}}
\renewcommand{\thefootnote}{\arabic{footnote}}

 \maketitle

\begin{abstract}
We briefly review the diffraction of quasicrystals and then give an elementary alternative proof of the diffraction formula for regular cut-and-project sets, which is based on Bochner's theorem from Fourier analysis. This clarifies a common view that the diffraction of a quasicrystal is determined by the diffraction of its underlying lattice.
To illustrate our approach, we will also treat a number of well-known explicitly solvable examples.
\end{abstract}

\section{Outline}

Quasicrystals are highly ordered rigid structures that are -- unlike crystals -- intrinsically non-periodic. Their discovery by diffraction experiments on certain rapidly cooled alloys \cite{S, I} in 1982 was a surprise, since a combination of the above two properties was regarded unphysical at that time. In fact quasicrystals have been recently found in nature \cite{Bi09, Bi15}. In 2011 the Nobel Prize in chemistry was awarded to Dan Shechtman for the discovery of quasicrystals.

A natural mathematical idealisation of quasicrystals are (regular) cut-and-project sets, i.e., projections of suitable lattice subsets from some higher-dimensional superspace.  The Bragg part in the diffraction of a cut-and-project set had been rigorously computed by Hof \cite{Hof1} in 1995. Pure point diffractivity of a cut-and-project set was proved by Schlottmann \cite{Martin2} however only in 2000 (in fact for a more general setting than Euclidean space).
It has been argued from the beginning that the diffraction of a cut-and-project set should be deducible from that of the underlying lattice. However a corresponding proof  appeared only in 2013 by Baake and Grimm \cite{BG2}%
\footnote{The cited monograph is a very useful compendium about mathematical quasicrystals and, more generally, on mathematics of aperiodic order.}
for Euclidean superspace. In fact the diffraction formula for a cut-and-project set is, in a certain sense, equivalent to the diffraction formula for the lattice in superspace, as has been shown recently \cite{RS2}. The diffraction formula also holds for non-regular cut-and-project schemes of extremal density \cite{BHS, KR}, and it may be used to describe the Bragg part in the diffraction of certain random cut-and-project sets \cite{BMRS, LR, CR}.

Whereas there are some excellent expository reviews on aperiodic order, see e.g.~\cite{BG1, BG3}, it is our intention to actually re-derive a substantial part of the above results with relatively modest mathematical prerequisites, thereby reflecting recent important developments in the theory. This article may thus serve as a gentle introduction to mathematical diffraction theory of cut-and-project sets. In fact our derivation is not restricted to Euclidean superspace. It covers non-Euclidean examples such as in \cite{BMS, BM, BMRS}, which are of recent mathematical interest as they may code number theoretic problems \cite{BKKL2015, BJL, BHS, KR}.
Our techniques might even be adapted to cover certain non-abelian superspaces that were considered very recently \cite{BHP2}. Also, the assumption of a lattice in superspace for the cut-and-project construction might be relaxed considerably, by considering some positive definite measure. In particular, this might be applied to describe certain modulations in lattices or quasicrystals. Our approach rests on Bochner's theorem from Fourier analysis.  We will thus assume some familiarity with harmonic analysis in Euclidean space, measure theory and integration on groups. Some relevant results and references are collected in the appendix.

In the following sections, we will recall diffraction theory and cut-and-project sets, and we will then discuss pure point diffractivity. In Section~\ref{ldiff}, we will prove the diffraction formula for a lattice beyond Euclidean space and then show in Section~\ref{sec:ms} how this gives rise to the diffraction formula for the cut-and-project sets. For concreteness, we will illustrate our approach to the diffraction formula on well-known one-dimensional examples: Fibonacci sets, the periodic doubling point set \cite{BM} and squarefree integers \cite{me73, bmp00}.  The latter two have a non-Euclidean internal space, and the last one is a non-regular model set of positive configurational entropy, yet with pure point diffraction. These examples have all been treated a number of times before, see also \cite{BG2}. A certain diffraction formula for the latter model already appears in \cite{NL92}. But our way of computing the diffraction is self-contained and partly different. It thus complements previous approaches.

\section{Diffraction experiments and quasicrystals}

\subsection{Diffraction experiments}
Diffraction experiments are performed in order to resolve the internal structure of matter \cite{LP}. In such an experiment, a specimen is hit by an x-ray or neutron beam, and the intensity of the wave resulting from the interaction of the beam with the specimen is recorded at a plate perpendicular to the specimen at large distance.
Figure~\ref{fig:diff} left shows the x-ray diffraction image of an AlMnPd alloy perpendicular to a direction of tenfold symmetry, such as in \cite{S}. In fact the alloy also admits directions of two-fold and six-fold symmetry, supporting the presence of an icosahedral symmetry, which is incompatible with translational invariance.
\begin{figure}[htb]
\begin{center}
\begin{minipage}[b]{0.4\textwidth}
\center{\epsfig{file=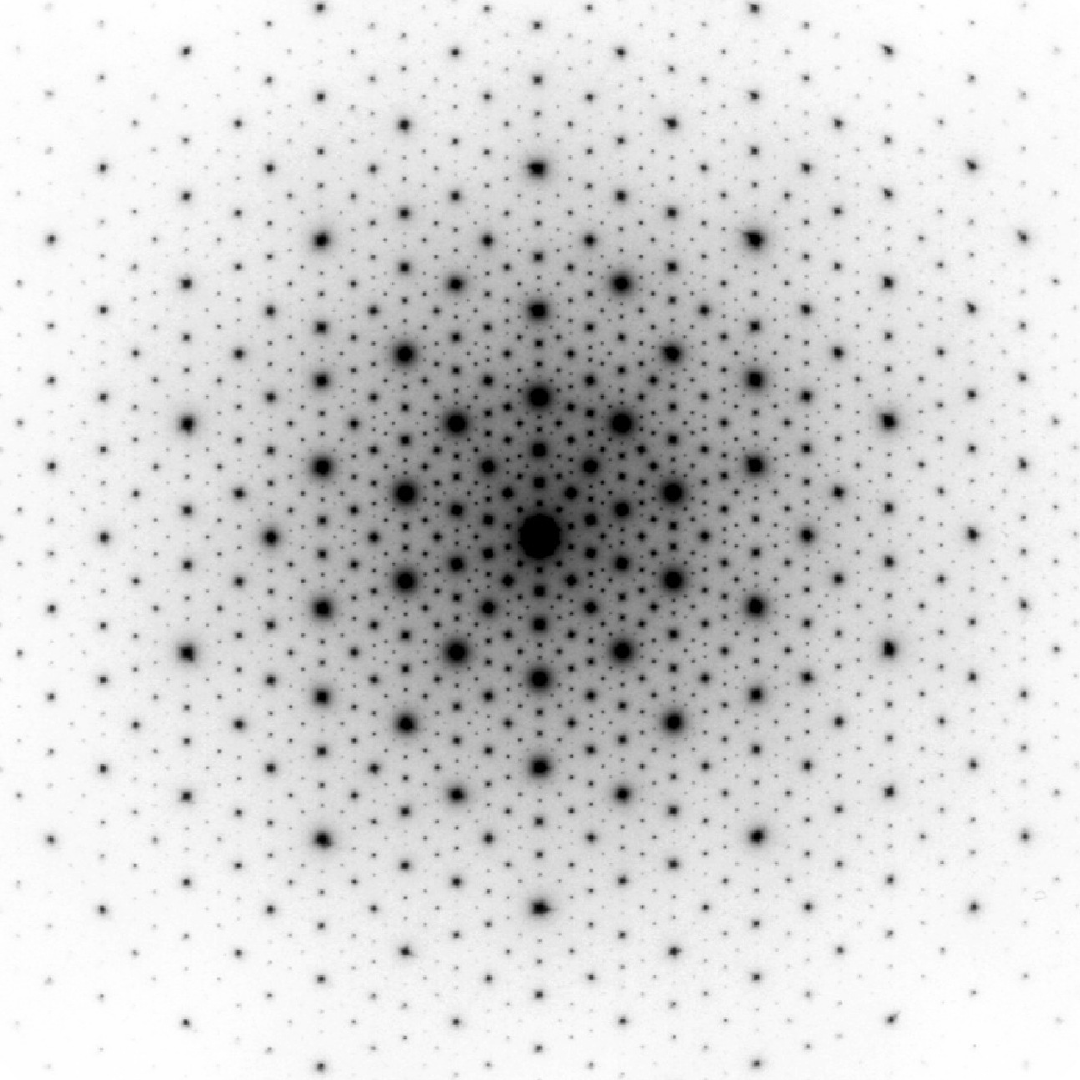,width=4cm, angle=90}}
\end{minipage}
\hfill
\raisebox{-0.2cm}{\begin{minipage}[b]{0.5\textwidth}
\center{\epsfig{file=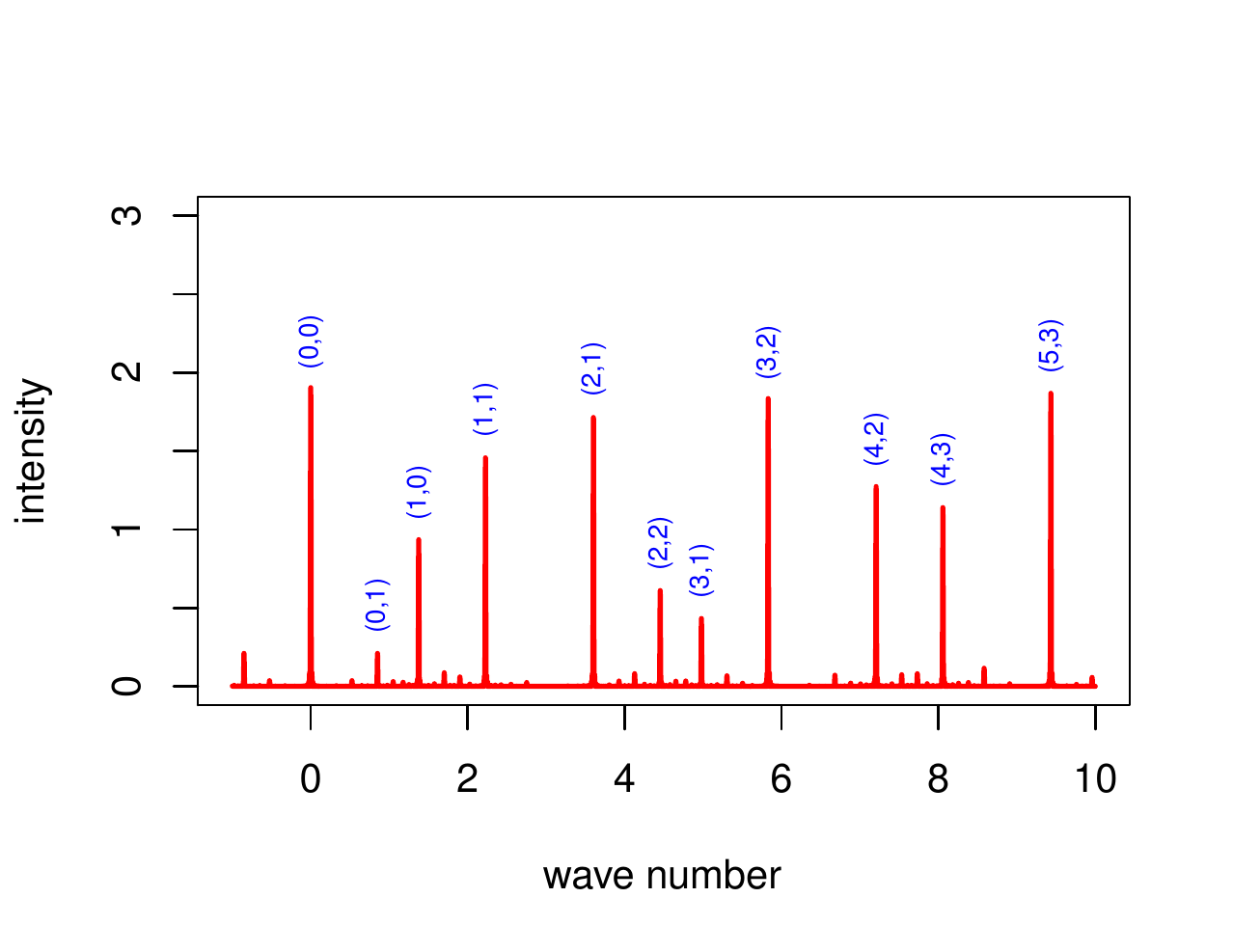, width=6cm}}
\end{minipage}}
\end{center}
\caption{left: Diffraction of an AlMnPd quasicrystal (copyright C.~Beeli). right: Diffraction of a Fibonacci set sample of size 100.}
\label{fig:diff}
\end{figure}
As for conventional crystals, the positions of the Bragg peaks can be indexed by integer linear combinations of a given set of fundamental vectors. For example,
the diffraction spots on the horizontal in the left panel of Figure~\ref{fig:diff}  may be indexed using two length scales whose ratio is the golden number $\tau=(1+\sqrt{5})/2$. But altogether the above example needs six fundamental vectors, instead of three for conventional crystals. Thus one is tempted to conjecture a dense set of Bragg peak positions, in contrast to the crystal case with an underlying lattice structure. Such ``unusual crystals'' were called quasicrystals by Levine and Steinhardt \cite{LS84}.

\subsection{Models for quasicrystals}\label{qcm}

Consider a (finite) specimen with atom positions $\Lambda\subset \mathbb R^3$. We assume for simplicity that all atoms are of the same type
and that  the incident beam is a monochromatic plane wave $\boldsymbol r\mapsto e^{-2\pi\imath \boldsymbol{k}_0\cdot\boldsymbol{r}}$. In kinematical%
\footnote{A realistic modelling at small wavelengths has to incorporate absorption cross ratios \cite{GW}.}
diffraction \cite[Sec.~II]{C}, the so-called structure factor $F(\boldsymbol s) =\sum_{\boldsymbol{p}\in \Lambda} e^{-2\pi\imath \boldsymbol{s}\cdot\boldsymbol{p}}$ describes the superposition of waves scattered by the atoms, and the diffraction intensity at large distance $\boldsymbol r$ is given by
\begin{displaymath}
I(\boldsymbol{r})=\frac{A}{|\boldsymbol{r}|^2}|F(\boldsymbol k-\boldsymbol k_0)|^2,
\end{displaymath}
where $\boldsymbol{k}=|\boldsymbol k_0| \cdot\frac{\boldsymbol r}{|\boldsymbol r|}$. Here $|\boldsymbol r|$ denotes the Euclidean norm of $\boldsymbol r$. By plotting the function $\boldsymbol s\mapsto |F(\boldsymbol s)|^2$ for a given model, one can thus study which point configurations match the diffraction in Figure~\ref{fig:diff} left.

Two fundamentally different types of model have been suggested. The first approach will be discussed below for the so-called cut-and-project construction, which yields a rigid point configuration from a lattice in higher dimensional space by some deterministic procedure. The second approach admits some randomness in the construction, and corresponding models have been suggested by Shechtman and Blech \cite{SB} and by Elser \cite{E}, the latter belonging to the class of random tilings \cite{Hen, RHHB, CR99}.  Also the nowadays prominent soft quasicrystals \cite{LD07} fall into that category. Due to their intrinsic randomness, diffraction of such models might have a continuous component. But a continuous component seems absent in the left panel of Figure~\ref{fig:diff}, apart from thermal fluctuations.

\subsection{Cut-and-project sets}\label{sec:cpsets}

The cut-and-project construction was developed by Kramer and Neri \cite{KN} before the discovery of quasicrystals, in order to systematically produce non-periodic space fillings with prescribed symmetries \cite{K1, Se}. Kramer and Neri were inspired by de Bruijn's so-called grid method \cite{dBr} for constructing Penrose tilings of the plane \cite{P}. The cut-and-project method is in fact equivalent to a multi-dimensional extension of the grid method \cite{GR}, compare \cite[Sec.~7.5.2]{BG2}. We illustrate it in Figure~\ref{fig:Fibo} for a Fibonacci set,
compare \cite[Ex.~7.3 and Sec.~7.5.1]{BG2}. Such sets may also be constructed recursively by a substitution rule, see e.g.~\cite{BHP} for a detailed study of this connection.
\begin{figure}[htb]
\begin{center}
\begin{minipage}[b]{\textwidth}
\center{\epsfig{file=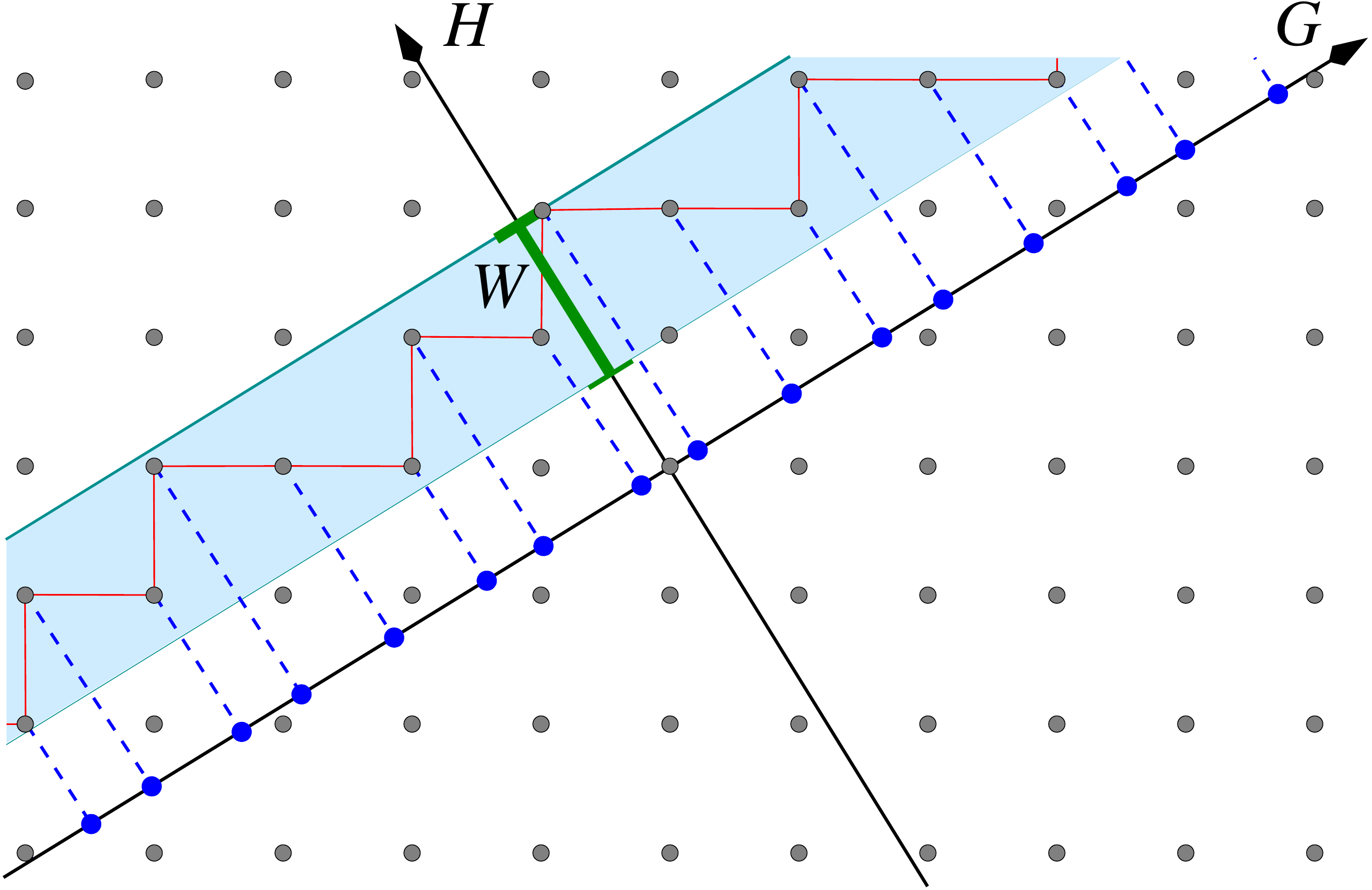,width=6cm}}
\end{minipage}
\end{center}
\caption{Cut-and-project construction of a Fibonacci set.}
\label{fig:Fibo}
\end{figure}

Given a lattice $\cL$ in superspace $G\times H$, a point set in physical space $G$ is obtained by projecting all lattice points inside a strip parallel to $G$. The space $H$ is called the \textit{internal space}, and the strip may be described by a \textit{window} $W\subseteq H$. Here $G$ has irrational slope $1/\tau$ with respect to the square lattice. The window is a half-open interval having the length of a projected unit square. Hence exactly two neighbour point distances emerge. The diffraction of a finite sample of the resulting point set is plotted in the right panel of Figure~\ref{fig:diff}. It closely resembles the diffraction of the full Fibonacci set. In fact the diffraction of a full Fibonacci set is the same for any shift of the underlying lattice and hence for any shift of the window. As indicated in the figure, the Bragg peak positions $(m,n)$ lie in  $c\cdot\mathbb Z[\tau]:=\{c\cdot(m\tau+n) : m,n\in\mathbb Z\}$. The value $c=\tau/\sqrt{2+\tau}$ is evaluated in Section~\ref{sec:fib}. It is perhaps surprising that this simple model matches the diffraction along the horizontal in the left panel of Figure~\ref{fig:diff}. For a partial explanation, note that icosahedral symmetry is a strong requirement for the cut-and-project construction. It is possible to implement such a symmetry in $G=H=\mathbb R^3$, and then in one-dimensional sections two length scales of ratio $\tau$ naturally appear \cite[Sec.~7.4]{BG2}.

Levine and Steinhardt \cite{LS84} extended de Bruijn's grid method to construct a three-dimensional space filling, which matches the diffraction in Figure~\ref{fig:diff} and the diffraction along sixfold directions. Elser \cite{E} suggested using the cut-and-project sets of Kramer and Neri for quasicrystal analysis. In fact at the same time the cut-and-project method was re-discovered by Kalugin, Kitaev and Levitov \cite{KKL} and by Duneau and Katz \cite{DK}, see also  the discussion in \cite{GR}.
It was later remarked by Lagarias \cite{L} that cut-and-project sets already appeared in Meyer's work on harmonious sets \cite{me72} as so-called model sets. Meyer's work has then been taken up and advocated particularly by Moody \cite{RVM3}, and this approach has been quite influential ever since.

\section{Pure point diffractivity}
There are two equivalent ways for computing the intensity function $\boldsymbol s \mapsto |F(\boldsymbol s)|^2$, which may be summarised in a so-called Wiener diagram.
\begin{displaymath}
\begin{CD}
\omega @>*>> \omega*\widetilde{\omega}\\
@V{\mathcal F}VV @VV{\mathcal F}V\\
{\widehat{\omega}} @>{|\cdot|^2}>> \widehat{\omega}\cdot \overline{\widehat{\omega}}
\end{CD}
\end{displaymath}
The atom positions are described by the (finite) Dirac comb $\omega=\sum_{\boldsymbol p\in \Lambda}\delta_{\boldsymbol p}=:\delta_\Lambda$. We often view the measure $\omega$ as a linear functional via $\omega(f)=\sum_{\boldsymbol p\in \Lambda}f(\boldsymbol p)$. The method described above appears in the lower left path of the diagram: First the structure factor is computed as the Fourier transform $\widehat \omega$, which we regard as a measure of the group $\widehat G$ dual to $G$, see Appendix~\ref{app:B}. As we may identify $\widehat{\mathbb R^d}$ and $\mathbb R^d$, see Appendix~\ref{app:LCA}, we may thus identify $\widehat\omega$ and the function $\boldsymbol s\mapsto F(\boldsymbol s)$ in this case, compare Remark \ref{rem:FT}~(iv). Next, the squared modulus  $\left| \widehat{\omega} \right|^2= \widehat{\omega}\cdot \overline{\widehat{\omega}}$ of the function $\widehat \omega$ is taken. Due to the convolution theorem, one might alternatively first compute the (finite) so-called autocorrelation measure $\omega * \widetilde{\omega}$, whose mass function is also called the Patterson function \cite{C}. Here the reflected measure $\widetilde{\omega}$ is defined as $\widetilde{\omega}(f)=\overline{\omega(\widetilde f)}$, and  $\widetilde f(x)=\overline{f(-x)}$.  The $*$-symbol denotes convolution of measures, see Appendix~\ref{app:B}, which corresponds to convolution of functions in this case. The Fourier transform of  $\omega * \widetilde{\omega}$ gives the intensity function.

For a (finite) specimen, the diffraction intensity is a continuous function. Pure point diffractivity is defined on the infinite idealisation of the specimen and arises as follows. In the limit of infinite sample size and after suitable normalisation of the intensity function, pure point diffraction will manifest itself in a discontinuous limiting function vanishing almost everywhere, up to some countable set of Bragg peak positions. For a mathematical description, we consider an infinite idealisation $\omega$ and compute diffraction on restrictions  $\omega_n$ to balls $B_n$ of radius $n$, which we may assume to be centered at the origin for simplicity.  We then take the limit $n\to\infty$ of the normalised diffraction intensities  $\frac{1}{\mathrm{\theta}(B_n)}\widehat{\omega_n}\cdot \overline{\widehat{\omega_n}}$.

Taking the limit after the Fourier transform is often computionally more difficult than taking the transform on an infinite object. Both operations commute when working with so-called positive definite measures, as the transform is continuous on positive definite measures, see \cite[Thm.~4.16]{BF} and \cite[Lemma~1.26]{MoSt}. In particular for a lattice Dirac comb $\omega=\delta_\cL$, which is positive definite, one may thus take the lower left path on the infinite lattice itself, which yields
\begin{displaymath}
\mathcal F\left(\delta_\cL\right)=\mathrm{dens}(\cL)\cdot \delta_{\cL^0}
\end{displaymath}
by the Poisson Summation Formula. Here $\cL^0$ is the lattice dual to $\cL$, and $\mathrm{dens}(\cL)$ denotes the density of lattice points, see Appendix~\ref{app:LCA}. We can read off that the diffraction is pure point as the limiting intensity measure is supported on the dual lattice and hence a point measure.  But also the (normalised) autocorrelation exists as the vague limit%
\footnote{Vague convergence means $\omega_n(f)\to \omega(f)$ for all continuous compactly supported functions $f$.
}
\begin{displaymath}
\gamma_\omega :=\omega \circledast \widetilde \omega :=\lim_{n\to\infty} \frac{1}{\mathrm{\theta}(B_n)} \omega_n*\widetilde{\omega_n}=\mathrm{dens}(\cL)\cdot \delta_{\cL}\ .
\end{displaymath}
Hence by continuity of the Fourier transform we have
\begin{displaymath}
\lim_{n\to\infty}\frac{1}{\mathrm{\theta}(B_n)}\widehat{\omega_n}\cdot \overline{\widehat{\omega_n}}= \lim_{n\to\infty}\mathcal F\left( \frac{1}{\mathrm{\theta}(B_n)} \omega_n*\widetilde{\omega_n}\right)=\mathcal F\left(\gamma_\omega\right)=\mathrm{dens}(\cL)^2\cdot \delta_{\cL^0}\ .
\end{displaymath}
This shows that the two approaches for  computing diffraction coincide when transforming the limits. It also shows that the diffraction of the infinite idealisation can be inferred from a finite sample if the sample size is sufficiently large. The lattice case will be discussed in the following section. We will analyse it beyond Euclidean space, as this setting is relevant for cut-and-project sets.

Dirac combs $\omega$ of cut-and-project sets need not be positive definite. However the positive definite autocorrelation measure will exist for so-called regular model sets. Hence in this case one can analyse diffraction using the upper right part in the Wiener diagram.  This will be discussed in Section~\ref{sec:ms}. In fact convergence of the finite sample diffraction measures is uniform in the center of the balls, which makes it possible to infer the diffraction experimentally from finite samples.  We will not consider the lower left part in the Wiener diagram in this article, as the  Fourier transform $\widehat \omega$ may only exist as a tempered distribution \cite{LO}. In fact, if $\omega$ is a cut-and-project set, the Fourier transform $\widehat \omega$ exists in measure sense only when $\omega$ is a small deformation of a fully periodic crystal \cite{LS2}. This complicates the mathematical analysis, especially beyond Euclidean internal space.

\section{Lattice diffraction}\label{ldiff}

In this section, we re-derive the diffraction formula for a lattice, and we show that the two approaches of deriving the diffraction formula coincide. A detailed analysis in the Euclidean setting is given in \cite[Sec.~9.2]{BG2}. We adopt the terminology and notation described in the appendix. Consider a lattice $\cL$ in some (compactly generated) LCA group $G$. Then its Dirac comb $\delta_\cL\in \mathcal M^\infty(G)$ is a translation bounded positive and positive definite measure, see Appendix \ref{app:B} for definitions. The same is true of the Dirac comb $\delta_{\cL^0}\in \mathcal M^\infty(\widehat G)$ of the dual lattice.

The following result is known as the Poisson Summation Formula for a lattice, see e.g.~\cite[Thm.~5.5.2]{Rei2}. It extends to Haar measures on closed subgroups \cite[Thm.~6.19]{BF}.

\begin{proposition}[lattice PSF]\label{psflattice} Assume that $\cL$ is a lattice in $G$. We then have
\[
\langle \delta_{\cL}, f\rangle =\mathrm{dens}(\cL) \cdot \langle \delta_{\cL^0}, \widecheck{f}  \rangle
\]
for all $f \in PK(G)$.
\end{proposition}

For the above statement, note that we have $\langle \delta_{\cL},f\rangle=\sum_{\ell\in\cL}f(\ell)$ as explained in the appendix. On the rhs, $\mathrm{dens}(\cL)$ is the density of points in $\cL$, see Appendix~\ref{app:LCA}, and $\widecheck{f}$ denotes the inverse Fourier transform of $f$, see Appendix~\ref{app:pdfctns}.
The set $PK(G)$, see Definition~\ref{def:pdfctn}, consists of all positive definite continuous compactly supported functions $f: G\to\mathbb C$. 

\begin{remark}
The measure $\mathrm{dens}(\cL) \cdot  \delta_{\cL^0}$ is uniquely determined by the above equations \cite[Thm.~2.2]{ARMA1}. It is called the Fourier transform $\widehat{\delta_\cL}\in \mathcal M^\infty(\widehat G)$ of $\delta_\cL\in \mathcal M^\infty(G)$.
\end{remark}

For the ease of the reader, we will give an elementary proof that is based on Bochner's theorem and on an explicit analysis of the periodicity properties of the lattice.

\begin{proof}
\emph{Step 1:} We prove that there exists a positive measure $\mu$ on $\widehat{G}$ such that for all $f \in PK(G)$ we have $\widecheck{f} \in L^1(\mu)$ and $\langle \delta_{\cL}, f  \rangle = \langle \mu, \widecheck{f} \rangle$.

Consider arbitrary $f \in PK(G)$ and note $f^\dag*\delta_{ \cL} \in PK(G)$ by Lemma~\ref{pdpk}. Here, $f^\dag$ denotes the function $f^\dag(x)=f(-x)$, see Appendix~\ref{app:pdfctns}. For convolution between an integrable function and a measure, see the end of 
Appendix~\ref{app:B}.1.
By Bochner's Theorem \ref{bochthm}, there exists a positive finite measure $\sigma_f$ such that
\[
f^\dag *\delta_{\cL}(x) = \int_{\widehat G} \chi(x) {\rm d} \sigma_f(\chi)=:\widecheck{\sigma_f}(x) \,
\]
 for all $x\in G$.
%
We can now define a positive linear functional $\mu$ on $C_c(\widehat G)$ as follows: given $\phi \in C_c(\widehat{G})$ and compact $K\supseteq \supp(\phi):=\overline{\{g\in G: |\phi(g)|\ne0\}}$,  pick $f \in PK(G)$ such that $\widecheck{f} (\chi) \ge 1$ on $K$, compare \cite[Prop.~2.4]{BF}.
Now, the definition
\[
\langle \mu, \phi \rangle = \langle \sigma_f,  \frac{\phi}{\widecheck{f}}\rangle
\]
does not depend on the choice of $f$. Indeed, for $g\in PK(G)$  we have $\widecheck{g}\cdot \sigma_f = \widecheck{f} \cdot\sigma_g$ by Remark~\ref{expdf2} (iii). Moreover, if $\phi \in C_c(G)$ is non-negative, we have $\langle \mu, \phi \rangle \geq 0$, and therefore $\mu$ is a measure by Lemma \ref{positive implies measure}.
Finally the definition of $\mu$ gives for any $f\in PK(G)$ that $\widecheck{f} \mu = \sigma_f$. Since $\sigma_f$ is a finite measure, we get that $\widecheck{f} \mu$ is finite and hence $\widecheck{f} \in L^1(\mu)$.
Therefore we have
\begin{displaymath}
\langle \delta_{\cL}, f \rangle   = f^\dagger *\delta_{\cL}(0) = \widecheck{\sigma_f}(0) =  \widecheck{\check{f}\mu }(0) = \langle \mu, \widecheck{f} \rangle \ .
\end{displaymath}
As $f\in PK(G)$ was arbitrary, this proves Step 1.

\noindent \emph{Step 2:} We show that $\mu$ is a Haar measure on $\cL^0$.

Let $f \in PK(G)$ be arbitrary. Then, for all $\chi \in \cL^0$ we have $\chi \delta_{\cL} =\delta_{\cL}$ and hence
\begin{displaymath}
\langle T_{\chi}\mu, \widecheck f \rangle = \langle \mu, T_{-\chi} \widecheck f \rangle
= \langle \mu,  \widecheck{\chi f} \rangle = \langle \delta_\cL, \chi f\rangle=
\langle \delta_\cL, f\rangle = \langle \mu, \widecheck f\rangle \ ,
\end{displaymath}
where $T_{\chi}$ denotes translation by $\chi$.
This shows that $\mu-T_{-\chi}\mu$ vanishes on a subset of $C_0(G) \cap L^1(\mu)$ which is dense in $C_0(G)$, the set of continuous functions on $G$ vanishing at infinity. Hence we have $\mu=T_{-\chi} \mu$. Therefore $\mu$ is $\cL^0$-invariant. Similarly consider $x \in \cL$. Then $T_x \delta_{\cL} =\delta_{\cL}$ and hence
\begin{displaymath}
\langle x\mu, \widecheck f \rangle = \langle \mu, x \widecheck f \rangle
= \langle \mu,  \widecheck{T_x f} \rangle = \langle \delta_\cL, T_x f\rangle=
\langle T_{-x}\delta_\cL, f\rangle = \langle \delta_\cL, f\rangle=\langle \mu, \widecheck f\rangle \ .
\end{displaymath}
Note that in $x\mu$ we use the function $x:\widehat G\to \mathbb C$ associated to $x\in G$, which is defined via $x(\chi)=\chi(x)$ for all $\chi\in \widehat G$.
The above equation shows that $x \mu =\mu$ for all $x\in \cL$, which implies $\supp(\mu)\subset \cL^0$. Indeed, take any $\chi\notin \cL^0$. Then there exists $x\in (\cL^0)^0=\cL$ such that $x(\chi)=\chi(x)\ne 1$, which implies $\mu(U_\chi)=0$ in some neighborhood $U_\chi$ of $\chi$. Therefore $\mu$ is supported on $\cL^0$ and is $\cL^0$-invariant. Thus $\mu$ is a Haar measure on $\cL^0$.

\noindent \emph{Step 3:} We evaluate the normalisation constant in $\mu=C \delta_{\cL^0}$ as
\[
C= C \delta_{\cL^0} (\{ 0\}) = \mu(\{0 \}) =\dens (\cL) \ .
\]

For the latter equality, choose an averaging sequence $(A_n)_{n \in \mathbb N}$ in $G$ and take $\varphi=\psi*\widetilde \psi$ where $\psi\in C_c(G)$ such that $\int_G \varphi=1$. We then define $f_n: G\to \mathbb R$ by
\begin{displaymath}
f_n=\varphi*\frac{1}{\theta_G(A_n)} 1_{A_n} \ ,
\end{displaymath}
where $1_A$ denotes the characteristic function of the set $A$.
Note that $\langle \delta_\cL, f_n\rangle$ is a smoothed version of the relative frequency of lattice points in $A_n$, and it is easy to see that $\langle \delta_\cL, f_n\rangle\to \mathrm{dens}(\cL)$ as $n\to\infty$.  We have $f_n\in \mathrm{span}(PK(G))$, the set of linear combinations of functions from $PK(G)$. This may be seen using the polarisation identity (see e.g.~\cite[Prop.~1.9.4]{MoSt} or \cite[Rem.~3.1.2]{P89}). One can also show $\widecheck{f_n}(\chi)\to \delta_{\chi, e}$ as $n\to\infty$, compare \cite[Lemma~3.14]{RS2}. Here $\delta_{\chi,e}$ equals one if $\chi$ is the trivial character, and zero otherwise.  Moreover $\widecheck \varphi$ is an integrable majorant of $\widecheck{f_n}$.  Now we can use Step 1 and dominated convergence to infer that $\langle \delta_\cL, f_n \rangle=\langle \mu, \widecheck f_n\rangle \to \mu(\{0\})$ as $n\to\infty$. For details of the argument, see also  \cite[Prop.~3.12]{RS2}, or \cite[Thm.~3.2]{Hof1} in the Euclidean setting.
\end{proof}
The autocorrelation of a lattice Dirac comb has a simple form.
\begin{proposition}[Lattice autocorrelation]\label{gamlat}
Let $\mathcal L$ be a lattice in $G$. Then the positive and positive definite measure $\delta_\cL\in \mathcal M^\infty(G)$ satisfies
\begin{displaymath}
\gamma_{\delta_\cL}=\delta_\cL \circledast {\widetilde \delta_\cL}=\mathrm{dens}(\cL) \cdot  \delta_\cL \ .
\end{displaymath}
\end{proposition}

\begin{proof}
Take an averaging sequence $(A_n)_{n\in\mathbb N}$ in $G$. With $\mathrm{card}(A)$ denoting the cardinality of the set $A$,  we then have
\begin{displaymath}
\begin{split}
\gamma_n&= \frac{1}{\theta_G(A_n)} \, \delta_\cL|_{A_n} * \widetilde{ \delta_\cL|_{A_n}}=
\frac{1}{\theta_G(A_n)} \, \delta_\cL|_{A_n} * \delta_\cL|_{-A_n}\\
&= \sum_{z\in \cL} \frac{\mathrm{card}(\cL \cap A_n \cap (z+A_n))}{\theta_G(A_n)} \delta_z
= \sum_{z\in \cL} \frac{\mathrm{card}(\cL \cap A_n)}{\theta_G(A_n)} \delta_z +o(1)\\
&= \mathrm{dens}(\cL) \cdot  \delta_\cL +o(1)
\end{split}
\end{displaymath}
as $n\to\infty$ with respect to vague convergence.
Hence $\gamma_{\delta_\cL}=\lim_{n\to\infty} \gamma_n=\mathrm{dens}(\cL) \cdot  \delta_\cL$.
\end{proof}

\begin{remark}[Lattice Wiener diagram]\label{genWien}
The above results can be summarised in a generalised Wiener diagram for a lattice Dirac comb.
\begin{displaymath}
\begin{CD}
\omega=\delta_\cL @>\circledast>> \gamma_\omega=\mathrm{dens}(\cL) \cdot \delta_\cL\\
@V{\mathcal F}VV @VV{\mathcal F}V\\
\widehat{\omega}=\mathrm{dens}(\cL)\cdot \delta_{\cL^0} @>{|\cdot|^2}>> \widehat{\gamma_\omega}=\mathrm{dens}(\cL)^2\cdot \delta_{\cL^0}
\end{CD}
\end{displaymath}
\medskip

The Wiener diagram expresses that the two approaches indicated in the previous section for computing the diffraction are equivalent.
This will no longer be the case for Dirac combs of non-periodic quasicrystals, when we regard the objects in the above diagram as measures.
\end{remark}

\section{Diffraction of regular model sets}\label{sec:ms}

We formalise the setting of Section~\ref{sec:cpsets}. Let $G=\mathbb R^d$ and let $H$ be a compactly generated LCA group.
We write $\pi^G: G\times H\to G$, $\pi^H: G\times H \to H$ for the canonical projections. Given a lattice $\cL$ in $G\times H$, the triple
 $(G, H, \cL)$ is called a \textit{cut-and-project scheme} if (i) $\pi^G$ is one-to-one on $\cL$, and if (ii) $\pi^H(\cL)$  is dense in $H$. Let us call these
two conditions the \textit{projection assumptions}.

\begin{remark}
\begin{itemize}
\item[(i)] Given $(G, H, \cL)$, we have that $(\widehat{G}, \widehat{H},  \cL^0)$ is also a cut-and-project scheme.  Indeed, $\pi^{\widehat G}$ is one-to-one on $\cL^0$ if and only if $\pi^H(\cL)$ is dense in $H$, and $\pi^{\widehat H}(\cL^0)$ is dense in $\widehat H$ if and only if $\pi^G$ is one-to-one on $\cL$. This is a consequence of  Pontryagin duality, see e.g.~\cite[Sec.~5]{RVM3}.
\item[(ii)] In the Euclidean setting $G=\mathbb R^d$ and $H=\mathbb R^n$, assume that $\cL$ is a rotated scaled copy of $\mathbb Z^{d+n}$. Then either of the two projection assumptions implies the other by duality, compare Figure~\ref{fig:Fibo}.
\end{itemize}
\end{remark}

\begin{definition}[model set] \rm  Let a cut-and-project scheme $(G, H, \cL)$ and a \textit{regular window} $W\subseteq H$ be given, i.e., a relatively compact and measurable set with non-empty interior such that $\theta_H(\partial W)=0$. Then $\oplam(W)=\pi^G(\cL\cap (G\times W))$ is called a \textit{regular model set}. If $W$ is relatively compact and measurable, then $\oplam(W)$ is called a \textit{weak model set}.
\end{definition}

\begin{remark}\label{rmsrem}
\begin{itemize}
\item[(i)] Since for a given weak model set $\oplam(W)$ we can replace $H$ by the group generated by $W$, the assumption that $H$ is compactly generated is no restriction.
\item[(ii)] Since we may pass from $H$ to $\overline{\pi^H(\cL)}$, even if the second projection assumption does not hold we can assume that $\pi^H(\cL)$ is dense in $H$ without loss of generality. A prominent example where the second projection assumption is violated is the Penrose point set, when projected from $G\times H=\mathbb R^{2+3}$ and $\cL$ a rotated copy of $\mathbb Z^5$, compare \cite{dBr} and \cite[Rem.~7.8]{BG2}. Another one-dimensional example is the period doubling point set, which will be discussed in Section~\ref{sec:pds}.
\end{itemize}
\end{remark}

 For $h:H\to\mathbb C$ bounded and compactly supported, consider the weighted Dirac comb $\omega_h\in \mathcal M^\infty(G)$ defined by
\begin{displaymath}
\omega_h= \sum_{(x,y)\in \cL} h(y) \delta_x\ .
\end{displaymath}
If $\pi^G$ is one-to-one on $\cL$, we may identify $\oplam(W)$ with $\omega_{1_W}$.
In the sequel we compute the diffraction of measures $\omega_h$ for suitable weight functions $h$.
The following theorem is the model set analogue of the lattice PSF. Its proof is a simple application of the lattice PSF. Recall that $PK(G)$ denotes the set of positive definite and continuous compactly supported functions $G\to\mathbb C$, compare Appendix~\ref{app:pdfctns}.

\begin{theorem}\label{maintheo1} Let $(G, H, \cL)$ be a cut-and-project scheme and $h \in PK(H)$. Then $\omega_{\widecheck h}\in \mathcal M(\widehat G)$ is a positive measure, and we have
\begin{displaymath}
\langle \omega_h, g \rangle =\dens(\cL) \cdot  \langle \omega_{\widecheck{h}}, \widecheck g \rangle
\end{displaymath}
 for all $g \in PK(G)$. This result holds without the projection assumptions on $(G,H,\cL)$.
\end{theorem}

\begin{remark}
\begin{itemize}
\item[(i)] The measure $\omega_{\widecheck{h}}$ in the above equation is uniquely determined \cite[Thm.~2.2]{ARMA1}. It is called the Fourier transform  $\widehat{\omega_h}\in\mathcal M^\infty(\widehat G)$ of $\omega_h\in\mathcal M^\infty(G)$.
\item[(ii)] Note that the measure $\omega_{\widecheck{h}}$ has typically dense support, i.e., its Dirac point measures with positive amplitude lie dense in $\widehat G$. This makes such dense Dirac combs interesting, and they have been systematically studied in \cite{CR, LR, NS11}.
\end{itemize}
\end{remark}

\begin{proof}
Fix arbitrary  $g \in PK(G)$ and $h \in PK(H)$. Denoting their pointwise product in $G\times H$ by $g\odot h$, we then have
\begin{displaymath}
\langle \omega_h, g\rangle =  \langle \delta_{\cL}, g\odot h \rangle=\mathrm{dens}(\cL)\cdot\langle \delta_{\cL^0}, \widecheck{g}\odot \widecheck{h} \rangle=\mathrm{dens}(\cL)\cdot\langle\omega_{\widecheck{h}}, \widecheck{g}\rangle.
\end{displaymath}
Here the first and the last equalities hold by definition. The second equality is Proposition~\ref{psflattice}, since $g \odot h \in PK(G \times H)$ by Example~\ref{expdf2} (v). Note that all terms above are nonnegative and finite.

Now, for each $f \in C_c(\widehat{G})$ we can find some $g \in PK(G)$ such that $|f| \leq \widecheck{g}$, compare \cite[Prop.~2.4]{BF}. Thus $|f| \odot \widecheck{h} \in L^1( \delta_{\cL^0})$. This allows us to define a linear functional $\omega_{\widecheck{h}} : C_c(\widehat{G}) \to \CC$ via $\omega_{\widecheck{h}}(f)  = \langle \delta_{\cL^0}, f\odot \widecheck{h} \rangle$,
which is positive by the positivity of $\widecheck{h}$ and of $\delta_{\cL^0}$, and hence a measure by Lemma \ref{positive implies measure}.
\end{proof}

The method of computing the normalisation constant $\mathrm{dens}(\cL)$ in Proposition~\ref{psflattice} extends to the following result. A more general statement relating to so-called Fourier-Bohr coefficients and a proof are given in \cite[Prop.~3.12]{RS2}.

\begin{lemma}\label{FBgen}
Assume that $\mu\in\mathcal M^\infty(G)$ is positive definite. Then
\begin{displaymath}
\widehat \mu(\{0\})=\lim_{n\to\infty} \frac{\mu(A_n)}{\theta_G(A_n)}=:M(\mu).
\end{displaymath}
\qed
\end{lemma}

The above result is the key ingredient in our proof of the density formula for weighted model sets.

\begin{theorem}[Density formula for weighted model sets]\label{dens2}  Let $(G, H, \cL)$ be a cut-and-project scheme and $h: H\to\mathbb C$ Riemann integrable. Then
\[
M(\omega_h)= \mathrm{dens}(\cL)\cdot \int_H h(y) \, {\rm d}\theta_H(y)\ .
\]
This result uses only the second projection assumption. In particular if $W \subseteq H$ is a regular window, we have
\[
M(\oplam(W))= \mathrm{dens}(\cL)\cdot \theta_H(W) \,.
\]
\end{theorem}

\begin{proof}
Consider first any $h\in PK(H)$.
As $\omega_{\widecheck{h}}$ is a measure, we have $\omega_{\widecheck{h}}(\{ 0 \}) =  \widecheck{h} (0)$. Here we used that $\pi^{\widehat G}|_{\cL_0}$ is one-to-one, which follows from the second projection assumption by Pontryagin duality. Moreover, as $\omega_{h}$ is a translation bounded measure and positive definite, we can apply Lemma~\ref{FBgen} to obtain
\begin{displaymath}
\widehat{\omega_h}(\{ 0\}) = \lim_{n\to\infty} \frac{ \omega_h(A_n)}{\theta_G(A_n)}\ .
\end{displaymath}
The claim for  $h\in PK(H)$ follows now from Theorem~\ref{maintheo1}. Next, let $h:H\to\mathbb R$ be any Riemann integrable function. For such $h$ the result follows from the above by approximation, as by the density of $\mathrm{span}(PK(H))$ in $C_c(H)$, there exist two functions $g_1,g_2 \in PK(G)$ such that $g_1 \leq h \leq g_2 $ and $\int (g_2 -g_1)\, {\rm d}\theta_H  \le \varepsilon$, for any $\varepsilon>0$. See e.g.~\cite[Thm.~4.15]{RS2} for details of the argument. In particular, the claim holds for regular $W$, as its characteristic function $1_W$ is Riemann integrable. The case of complex valued $h$ is treated by analysing the real and complex parts separately.
\end{proof}

Next we recall the autocorrelation formula of a Dirac comb $\omega_{1_W}$ for a regular window $W$, see e.g.~\cite{Hof1, Martin2, BM, BG2}. For the ease of the reader, we review the computation. It is similar to that for a lattice in Proposition~\ref{gamlat}, but it additionally uses the density formula Theorem~\ref{dens2}. Note that the first projection assumption is not needed in the derivation, compare \cite[Prop.~5.1]{RS2}.

\begin{proposition}[Autocorrelation of regular model sets]\label{thm:ac}
Let $(G,H,\cL)$ be a cut-and-project scheme and consider the Dirac comb $\omega_{1_W}\in\mathcal M^\infty(G)$ for the regular window $W\subseteq H$. Then $\omega_{1_W}$ has a unique autocorrelation measure $\gamma=\omega_{1_W}\circledast \widetilde{\omega_{1_W}}\in\mathcal{M}^\infty(G)$ which is given by
\begin{displaymath}
\gamma=\mathrm{dens}(\cL)\cdot\omega_{1_W*{\widetilde{1_{W}}}} \,.
\end{displaymath}
\end{proposition}

\begin{proof}
Take an averaging sequence $(A_n)_{n\in\mathbb N}$ in $G$.
The autocorrelation of $\omega_{1_W}\in\mathcal M^{\infty}(G)$ is defined as the vague limit of the finite autocorrelation measures $\gamma_n$ given by
 \begin{displaymath}
\gamma_n= \frac{1}{\theta_G(A_n)} \, \omega_{1_W}|_{A_n}*\widetilde{\omega_{1_W}|_{A_n}}
= \frac{1}{\theta_G(A_n)}\, \omega_{1_W}|_{A_n}*\omega_{\widetilde{1_W}}|_{-A_n}
=\sum_{(z, z')\in \cL} \eta_n(z')\delta_z\ .
\end{displaymath}
Here $|_{A_n}$ denotes restriction to $A_n$, and $\eta_n(z')$ is given by
 \begin{displaymath}
\begin{split}
\eta_n(z')&=  \frac{1}{\theta_G(A_n)} \sum_{(x,x') \in \cL \cap (A_n\cap(z+A_n)\times H)} 1_W(x')1_W(x'-z')\\
&= \frac{1}{\theta_G(A_n)} \sum_{(x,x')\in \cL\cap (A_n\times H)} 1_W(x')1_W(x'-z')+o(1)\\
&=\frac{1}{\theta_G(A_n)} \omega_{1_W(\cdot)1_W(\cdot\, -z')}(A_n) + o(1)
\end{split}
\end{displaymath}
asymptotically as $n\to\infty$. Here we used for the estimate that $A_n$ is a ball, $W$ is relatively compact and $\cL$ is uniformly discrete.
Since the function $y\mapsto 1_W(y) 1_W(y-z')$ is Riemann integrable on $H$, we can apply the density formula for weighted model sets and obtain
\begin{displaymath}
\begin{split}
\eta(z')&=\lim_{n\to\infty} \eta_n(z')=\mathrm{dens}(\cL)\cdot \int_H 1_W(y)1_W(y-z')\,{\rm d}\theta_H(y)\\
&=\mathrm{dens}(\cL)\cdot (1_W*\widetilde{1_W})(z').
\end{split}
\end{displaymath}
Here we used the second projection assumption.
Since $\supp(\omega_{1_W*\widetilde{1_W}})$ is uniformly discrete as $1_W*\widetilde{1_W}$ is compactly supported, this implies that $\gamma_n$ converges vaguely to $\gamma$.
\end{proof}

Combining this result with Theorem~\ref{maintheo1}, we arrive at the diffraction formula for regular model sets, see e.g.~\cite{Martin2, BM, BG2}. Note that due to regularity of the window we may replace $W$ by $\overline W$, as this does not affect the autocorrelation coefficients $\eta(z')$, see the above proof.
\begin{theorem}[diffraction formula for regular model sets]\label{diff rms} Consider the Dirac comb $\omega_{1_W}$ for some regular window $W\subseteq H$ in some  cut-and-project scheme $(G, H, \cL)$. Then $\omega_{1_W}$ has autocorrelation $\gamma$ and diffraction $\widehat\gamma$ given by
\[
\gamma = \dens(\cL) \cdot \omega_{1_{\overline{W}} *\widetilde{1_{\overline{W}}}}\ , \qquad \widehat{\gamma}= \dens(\cL)^2 \cdot \omega_{|\widecheck{1_{\overline{W}}}|^2} \,.
\]
\qed
\end{theorem}

\begin{remark}
\begin{itemize}
\item[(i)] The diffraction formula is reminiscent of the Wiener diagram for a lattice, see Remark~\ref{genWien}. Indeed, one may first compute the Fourier transform of the window in internal space and then its squared modulus. The latter is then evaluated on the Fourier module, which is the dual of the underlying lattice projected to the direct space.
\item[(ii)] The diffraction formula might no longer be valid if the second projection assumption is violated.
\end{itemize}
\end{remark}

\begin{remark}[maximal density implies pure point diffraction] \label{diff maximal dens}
Any weak model set $\oplam(W)$ satisfies the inequality
\[
\dens(\oplam(W)) \le \dens(\cL) \cdot \theta_H(\overline{W}) \,,
\]
which can be proved by approximating $W$ from above using regular windows \cite[Prop.~3.4]{HR15}. We say that a weak model set has maximal density if we have equality in the above expression, for some fixed averaging sequence. Regular model sets are of maximal density by Theorem \ref{dens2}. The approximation argument can be used to show that Theorem~\ref{diff rms} even holds for weak model sets of maximal density \cite[Cor.~6]{BHS}, compare also \cite[Sec.~3.3.2]{KR}.
\end{remark}

\section{Examples in one dimension}

\subsection{Fibonacci sets \cite[Sec.~9.4.1]{BG2}}\label{sec:fib}

The example in Section~\ref{qcm} has $G=H=\mathbb R$, and the lattice $\cL$ is a rotated copy of $\mathbb Z^2$. The window is an interval of length $1+\tau$. We have $\mathrm{dens}(\cL)=1$, and it is not difficult to show that $\pi^G(\cL)=\pi^H(\cL)=\frac{1}{\sqrt{2+\tau}}\mathbb Z[\tau]$.  Due to our choice of dual groups we have $\widehat G=\widehat H=\mathbb R$ and $\cL^0=\cL$. By a simple calculation, we evaluate the diffraction measure as
\begin{displaymath}
\widehat{\gamma}= \omega_{|\widecheck{1_W}|^2}
=\sum_{k\in \pi^{\widehat G}(\cL^0)} \left(\frac{\sin\left(\pi \frac{1+\tau}{\sqrt{2+\tau}} k^\star\right)}{\pi k^\star}\right)^2 \, \delta_k\ ,
\end{displaymath}
where $\left(\frac{1}{\sqrt{2+\tau}}(n+m\tau)\right)^\star=\frac{1}{\sqrt{2+\tau}}(n\tau-m)$. We can read off the density $\frac{1+\tau}{\sqrt{2+\tau}}=1.894427\ldots$ of the Fibonacci set from the intensity at the origin. The diffraction measure shows that the Bragg peak positions indeed lie dense in $\widehat G=\mathbb R$, as conjectured from the diffraction picture in the right panel of Figure~\ref{fig:diff}. In this example we also see that the linear functional $\omega_{\widecheck{ 1_W}}$ cannot be the Fourier transfrom of $\omega_{1_W}$ as a measure, as the former is not translation bounded, compare Remark~\ref{rem:FT} (ii). Indeed, since $x\mapsto h(x)=\sin(x)/x\notin L^1(\mathbb R)$, there is a sequence of non-negative functions $h_n\in C_c(H)$ such that $h_n\le |h|$ and $\int_H h_n\, {\rm d}\theta_H\to\infty$ as $n\to\infty$.
Therefore, if we assume by contradiction that $\omega_h$ would be a translation bounded measure, we would get that
\begin{displaymath}
\limsup_m \frac{|\omega_h|(A_m)}{\theta_G(A_m)}\ge M(\omega_{h_n})=\int_H h_n \, {\rm d}\theta_H \to\infty \qquad (n\to\infty)\ ,
\end{displaymath}
which contradicts the fact that $\omega_h$ is translation bounded.

\subsection{Squarefree integers}\label{sec:sfi}

An integer $n$ is squarefree if it does not contain a square, i.e., if $n \mod p^2\ne0$ for every prime $p$. We recall that the squarefree integers $\mathcal S$ are a weak model set, see  \cite{me73, bmp00, BM} and \cite[Ex.~10.3]{BG2}. We use  the cut-and-project scheme $(G,H,\mathcal L)$ where  $G=\mathbb R$, and where $H$ is the compact group given by
\begin{displaymath}
H=\prod_{p\in \mathcal P} \mathbb Z/ p^2\mathbb Z\ ,
\end{displaymath}
with $\mathcal P$ the set of primes \cite[Sec.~5a]{sing07}. This setting simplifies previous diffraction computations \cite{bmp00} as it avoids adelic internal spaces.
It is also well suited to study relations to $B$-free systems \cite{BKKL2015}, compare \cite{BHS, KR}.

We write $n^\star=(n \mod p^2)_{p\in \mathcal P}$ for $n\in\mathbb Z$ and note that $\cL=\{(n, n^\star): n\in\mathbb Z\}$ is a lattice in $G\times H$. Indeed it is a group, and  discreteness and relative denseness follow as $\mathbb Z$ is a lattice in $G$ and as $H$ is compact. Obviously $\pi^G$ is one-to-one on $\cL$, and $\pi^H(\cL)$ is dense in $H$, as can be seen from the Chinese remainder theorem. Note that $\ZZ=\oplam(H)$ is a regular model set, hence
\begin{displaymath}
1 = \dens(\ZZ) = \dens(\cL) \cdot \theta_{H}(H)=\dens(\cL) \,.
\end{displaymath}
For the squarefree integers, we have $\mathcal S=\oplam(W)$ with window
\begin{displaymath}
W=\prod_{p\in\mathcal P} \left(\mathbb Z /p^2\mathbb Z \setminus \{0_p\}\right) \ .
\end{displaymath}
The Haar measure of the window $W$ is given by \cite{BHS}
\begin{displaymath}
\theta_H(W)= \prod_{p\in \mathcal P} \left(1-\frac{1}{p^2}\right)=\frac{1}{\zeta(2)}>0 \ .
\end{displaymath}
The density of $\mathcal S$ exists when averaging on intervals centered in 0.  An explicit non-trivial calculation which we omit, see e.g.~\cite[Prop.~11]{bmp00}, reveals that
$\mathrm{dens}(\Lambda)= \theta_H(W)$, which means that $\mathcal S$ has maximal density. Hence $\mathcal S$  is pure point diffractive\footnote{
In fact $\mathcal S$ is a non-regular model set:  Since no component of $W$ is maximal, the compact window has empty interior.  We thus have $W=\partial W$ and hence $\theta_H(\partial W)>0$. This  also indicates that $\mathcal S$ has positive configurational entropy \cite{PH, HR15}, which might seem to contradict pure point diffractivity. But $\mathcal S$ has in fact zero Kolmogorov-Sinai entropy with respect to the natural pattern frequency measure \cite{BH, BHS, KR}.
}
by Remark~\ref{diff maximal dens}. The Fourier transform of the window is readily computed by exploiting the product structure of $H$, $W$ and the characters. We get
\begin{displaymath}
\widecheck{1_{W}}(\ell)= \int_H 1_W(h)\overline{\chi_\ell(h)}\,{\rm d}h= \prod_{p\in \mathcal P} \left(\delta_{0, \ell_p}-\frac{1}{p^{2}} \right)=(-1)^{|\ell|} \frac{1}{\zeta(2)} \prod_{\substack{p\in\mathcal P\\ \ell_p\ne 0}}\frac{1}{p^2-1}\ ,
\end{displaymath}
where $\ell=(\ell_p)_p\in \widehat{H}$ and $|\ell|=\sum_p |\ell_p|<\infty$. Next, we find a parametrisation for the dual lattice  $\cL^0\subset \widehat G\times \widehat H$ of $\cL$. By definition we have
\begin{displaymath}
\cL^0 = \{ (x,\ell) \in \widehat{G} \times \widehat{H} : \exp\left(- 2 \pi i  \, x \cdot n \right) \cdot \chi_\ell( n^\star ) =1 \text{ for all } n \in \ZZ \} \,.
\end{displaymath}
The annihilation condition on the characters is equivalent to
\begin{displaymath}
x \cdot 1 + \sum_{p\in\mathcal P} \frac{\ell_p \cdot 1}{p^2} \in \ZZ \ ,
\end{displaymath}
where we replaced $n$ by $1$ without loss of generality. We thus have
\begin{displaymath}\label{eq:cL0}
\cL^0= \left\{ \left(k-\sum_{p\in\mathcal P} \frac{\ell_p}{p^2}, \ell\right): k\in \mathbb Z, \ell=(\ell_p)_p \in \widehat H  \right\} \ .
\end{displaymath}
A moment's reflection reveals that the Fourier module $\pi^{\widehat G}(\cL^0)$ consists of all rationals with cubefree denominator. We thus get for the diffraction measure
\begin{displaymath}
\widehat \gamma= \omega_{|\widecheck{1_W}|^2}=\sum_{r\in\pi^{\widehat G}(\cL^0)} I_r \, \delta_r\ , \qquad I_{r}= \frac{1}{\zeta(2)^2} \prod_{p|q}\frac{1}{(p^2-1)^2}\ ,
\end{displaymath}
where  $r=m/q$ with $\mathrm{lcd}(m,q)=1$ and $q$ cubefree. The left panel of Figure \ref{fig:diffperiod} shows the diffraction of a squarefree integer sample on a logarithmic intensity scale, together with labels $(m,q)$ on the high intensity Bragg peaks. Note that the diffraction measure is $\mathbb Z$-periodic.

\begin{figure}[htb]
\begin{center}
\begin{minipage}[b]{0.45\textwidth}
\center{\epsfig{file=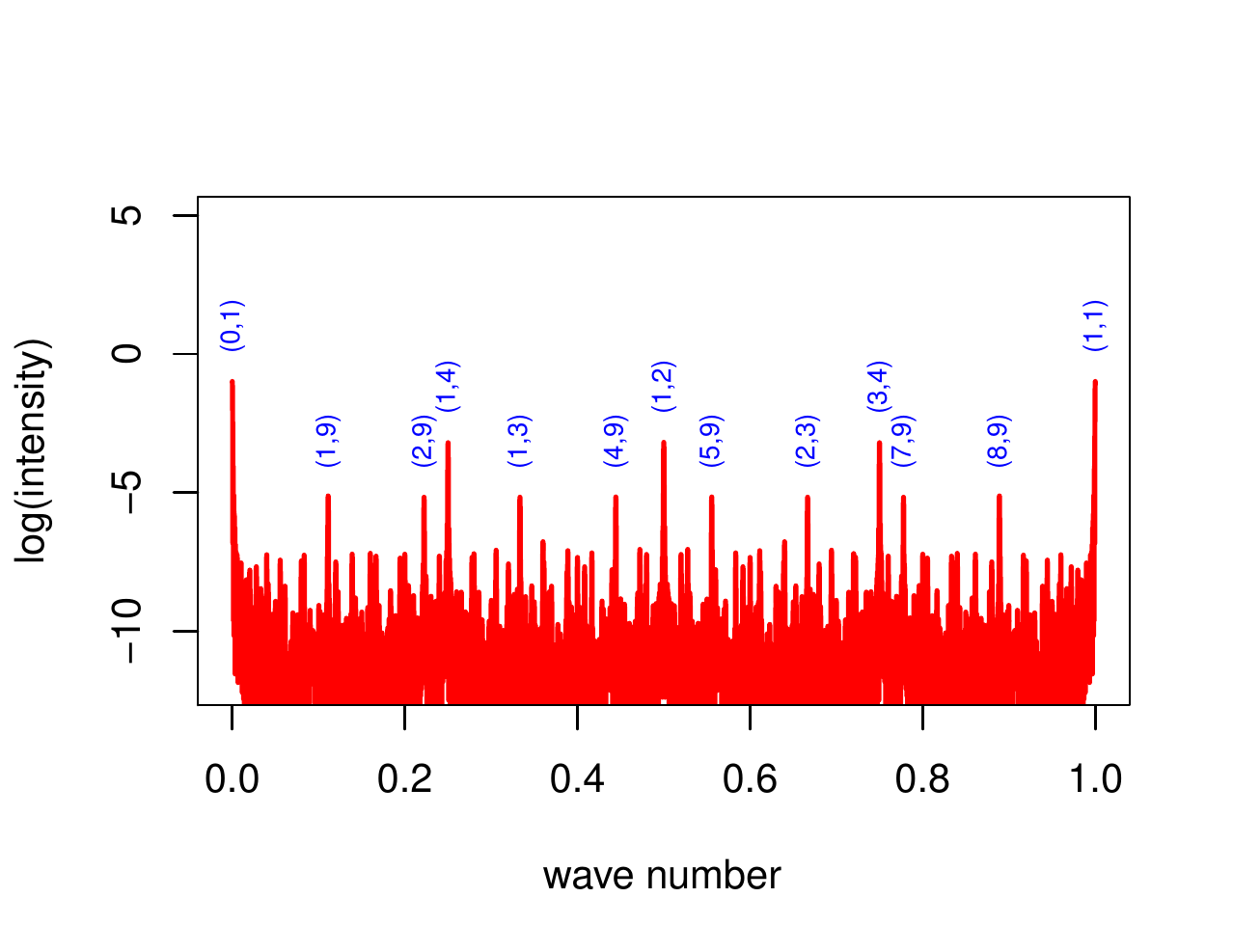, width=6.5cm}}
\end{minipage}
\hfill
\begin{minipage}[b]{0.45\textwidth}
\center{\epsfig{file=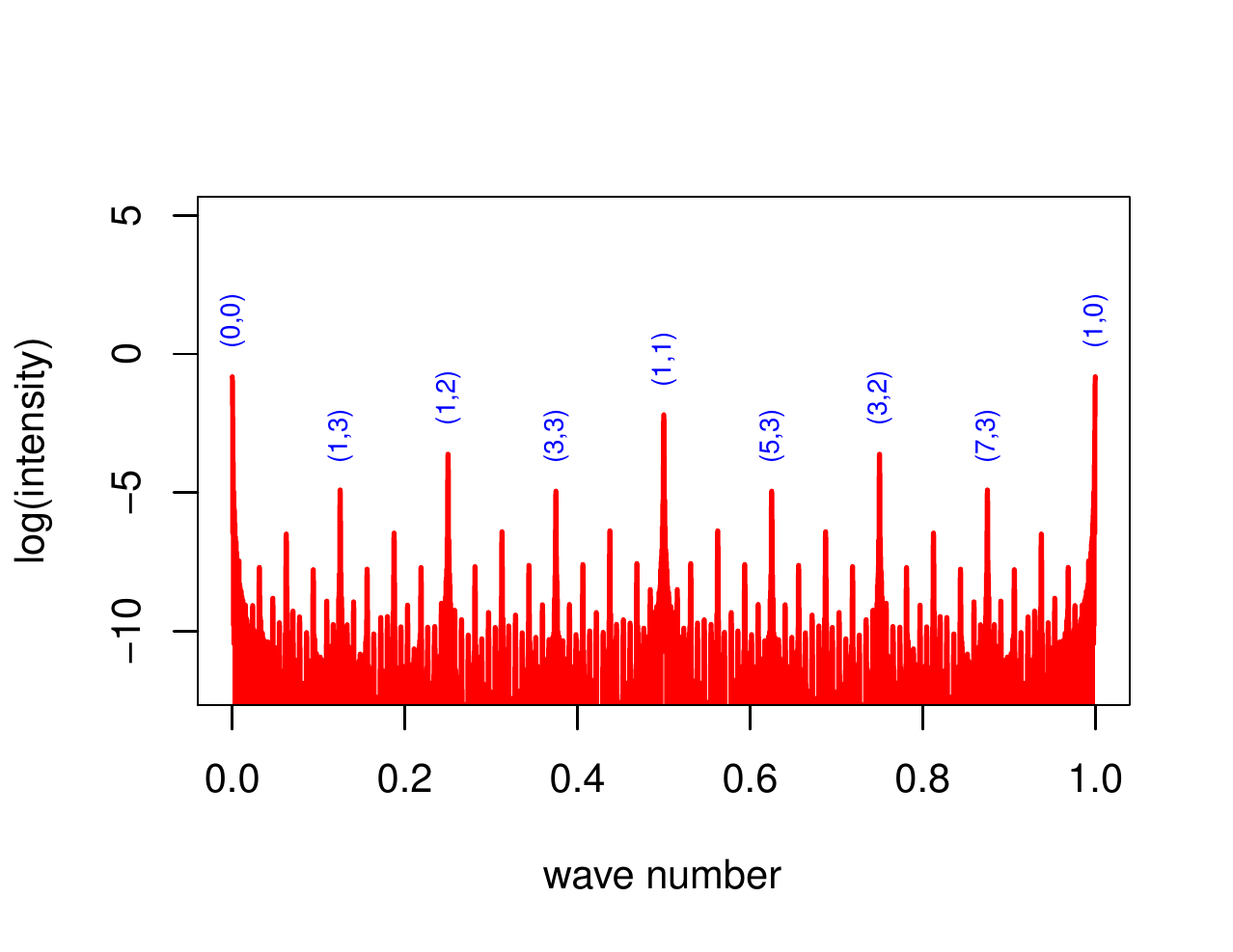, width=6.5cm}}
\end{minipage}
\end{center}
\caption{Diffraction of a size 1500 sample of squarefree integers (left) and  of the period doubling set (right).}
\label{fig:diffperiod}
\end{figure}

\subsection{Period doubling set \cite[Sec.~4.5.1]{BG2} }\label{sec:pds}

Let us consider $\mathcal D\subset \mathbb Z$ given by
\begin{displaymath}
\begin{split}
\mathcal D &= \bigcupdot_{j\ge 0} \left(2\cdot 4^j\mathbb Z+(4^j-1)\right) \,.
\end{split}
\end{displaymath}
It is a union of lattice cosets, which gives rise to a so-called limit periodic point set. The set $\mathcal D$ can be described as a so-called regular Toeplitz sequence \cite{BJL} and can be generated from the so-called period doubling substitution, see \cite[Ch.~8]{BM} for details. The limit periodicity can be used to compute the diffraction of $\mathcal D$ using the PSF on lattice approximants and by then taking the limit, as indicated in \cite[Ch.~8]{BM}.  Here we give a complete derivation based on the diffraction formula for regular model sets.  We split our calculation into several steps.

\emph{Step 1 (cut-and-project scheme and windows):}
A natural candidate of a cut-and-project scheme $(G,H,\mathcal L)$ for $\mathcal D$ is $G=\mathbb R$ and
$H=\prod_{j\in J} \mathbb Z/ 2^j \mathbb Z$, where $J=\mathbb N$. As in the previous example we argue that $\cL=\{(n, n^\star): n\in\mathbb Z\}$
is a lattice in $G\times H$, where now $n^\star=(n_j)_{j\in J}=  (n \mod 2^j)_{j\in J}$. We have $\mathcal D=\oplam(W)$ with $W=\bigcup_{j\in J} W_j$, where the clopen set $W_j$ equals $H$ with its $(2j+1)^{th}$ component replaced by $\{4^{j}-1\}\subseteq \mathbb Z/ (2\cdot 4^j)\mathbb Z$. In particular, $W$ is open in the infinite product topology.
Note that $\pi^H(\cL)$ is not dense in $H$, as the numbers $\{2^j:j\in J\}$ are not pairwise coprime. Consider thus
\begin{displaymath}
H'=\overline{\pi^H(\cL)}=\{ (h_j)_j: h_{j} = h_{j+1} \mod 2^j \text{ for all }j\in J\}\subset H
\end{displaymath}
together with the probability Haar measure $\theta_{H'}$ on the compact group $H'$.
With $\cL'=\cL\cap (G\times H')$ and $W'=W\cap H'$, the cut-and-project scheme $(G, H', \cL')$ satisfies $\cD = \oplam(W')$. Arguing as above, we find that
$\cL'$ has density $1$ in $G \times H'$. Observe that for $W_j'=W_j\cap H'$ the shifted lattice $\oplam(W'_j)=2 \cdot 4^j \ZZ+ (4^j-1)$ is a regular model set of density $\frac{1}{2 \cdot 4^j}$. Hence the density formula for regular model sets yields $\theta_{H'}(W_j')=\frac{1}{2 \cdot 4^j}$.

\emph{Step 2 (regularity of the window):}
As the sets $W_j'$ are clopen, we have for the boundary $\partial W'$ of $W'$ the estimate
\begin{displaymath}
\partial W'\subseteq \partial \left(\bigcup_{j\le n} W_j' \cup  \bigcup_{j>n} W_j' \right)
\subseteq \partial \left( \bigcup_{j>n} W_j' \right) \subseteq \overline{\bigcup_{j>n} W_j'}
\end{displaymath}
for any $n\in\mathbb N$. Noting that $h=(h_k)_k\in W_j'$ implies that $h_k=-1 \mod 2^k$ for $k<j$, we infer that $\partial W'\subseteq \{(-1)^\star\}$.  In particular, $W'$ is regular.  We can thus apply Theorem~\ref{diff rms} to compute the period doubling diffraction. In fact we have $\partial W' = \{(-1)^\star\}$, as is easily seen.

\emph{Step 3 (dual cut-and-project scheme):} Recall that $\widehat H=\bigoplus_{j\in J} \ZZ/2^j\ZZ$.  By Pontryagin duality we have $\widehat{H'}=\widehat H/(H')^0$, where the annihilator $(H')^0$ of $H'$ in $\widehat H$ is readily computed as
\begin{displaymath}
(H')^0= \left\{  \ell=(\ell_j)_{j\in J} \in \widehat{H} : \sum_{j\in J} \ell_j/2^j\in \mathbb Z \right\} \ ,
\end{displaymath}
compare the previous example.
Next, we find the dual lattice $\cL'^0\subseteq \widehat G\times \widehat{H'}$. A calculation very similar to that in the previous example yields
\begin{displaymath}
\cL'^0=\left\{\left(k+\sum_{j\in J} \frac{\ell_j}{2^j}, -\ell+(H')^0\right): k\in\mathbb Z, \ell=(\ell_j)_{j\in J}\in \widehat H\right\} \ .
\end{displaymath}
We observe that the Fourier module $\pi^{\widehat G}(\cL'^0)$ is given by
\begin{equation}\label{eq:zhalf writing}
\pi^{\widehat G}(\cL'^0)=\mathbb Z[\frac{1}{2}]= \left\{\frac{m}{2^r}: (r=0, m=\mathbb Z) \text{ or } (r\ge 1, m \text{ odd}) \right\} \ ,
\end{equation}
where the latter expression is chosen for a unique parametrisation \cite[Eqn.~(23)]{BM}. Defining $\psi :  \ZZ[\frac{1}{2}] \to \widehat{H'}$ by $\psi(x)=(\psi(x)_j)_{j\in J}+(H')^0$ where $\psi(m/2^r)_r= m \mod 2^r$ and $\psi(x)_j=0$ otherwise, we can thus write
\begin{displaymath}
\cL'^0=\left\{ (x, \psi(x)) : x \in \ZZ[\frac{1}{2}] \right\}\ .
\end{displaymath}
Hence $\psi$ is what is sometimes called the star map for $(\widehat G, \widehat{H'}, \cL'^0)$.

\emph{Step 4 (transform of the window):}
Note that by dominated convergence we can write
\begin{displaymath}
\widecheck{1_{W'}}(\psi(x))=\int_{W'} \chi_{\psi(-x)}(h) \, {\rm d}h = \sum_{j\in J} \int_{W'_j} \chi_{\psi(-x)}(h) \, {\rm d}h = \sum_{j\in J} \widecheck{1_{W_j'}}(\psi(x)) \ .
\end{displaymath}
Fix $x=m/2^r$. Now, if $2j+1 <r$, then $r>1$ and hence $m$ is odd. As $2j+1 \leq r-1$, an easy computation shows $(2^{r-1})^\star+W'_j=W'_j$. Therefore, by the invariance of the Haar measure we get
\begin{displaymath}
\begin{split}
\widecheck{1_{W'_j}}(\psi(x))&= \int_{W'_j} \chi_{\psi(-x)}(h) \, {\rm d}h
=  \int_{W'_j-(2^{r-1})^\star} \chi_{\psi(-x)}(h) \, {\rm d}h \\
&=  \int_{W'_j} \chi_{\psi(-x)}(h+(2^{r-1})^\star) \, {\rm d}h  =
 \int_{W'_j} \exp\left(-2 \pi i\, (-x)\cdot (h_{r} +2^{r-1}) \right) \, {\rm d}h \\
&=  \int_{W'_j}  \exp\left(-2 \pi i \frac{(-m) \cdot 2^{r-1}}{2^r} \right)  \exp\left(-2 \pi i \, (-x) \cdot h_{r} \right) {\rm d}h \\
&=  \int_{W'_j}  - \chi_{\psi(-x)}(h) \, {\rm d}h = -\widecheck{1_{W'_j}}(\psi(x)) \,.
\end{split}
\end{displaymath}
This shows that for $2j+1<r$ we have $\widecheck{1_{W'_j}}(\psi(x)) =0$.
Also, for all $2j+1 \geq r$ we can write $x= \frac{m}{2^r}= \frac{m\cdot2^{2j+1-r}}{2 \cdot 4^j}$, and this particular form makes the integration over $W'_j$ easier to  evaluate:
\begin{displaymath}
\begin{split}
\widecheck{1_{W_j'}}(\psi(x))&= \frac{1}{2 \cdot 4^j} \exp\left(-2 \pi i \, (-x)\cdot (4^j-1) \right)  \\
&= \frac{1}{2 \cdot 4^j} \exp\left(-2 \pi i \frac{(-m)\cdot2^{2j+1-r}}{2 \cdot 4^j}\cdot 4^j\right) \cdot\exp\left(-2 \pi i \, x\right)  \\
&= \frac{1}{2 \cdot 4^j}\cdot (-1)^{m\cdot2^{2j+1-r}} \cdot \exp(-2 \pi i \, x)   \ .
\end{split}
\end{displaymath}
Note here that when $2j+1 > r$ we have $(-1)^{m \cdot 2^{2j+1-r}}=1$. Therefore, all the terms of the form $(-1)^{m \cdot 2^{2j+1-r}}$ are one, excepting when $r$ is odd and $j=\frac{r-1}{2}$, disappear from the computation. Because of this, we split the problem in the cases $r$ odd and $r$ even.
For $r=2k$  even we compute
\begin{displaymath}
 \sum_{2j+1 \geq 2k} \frac{1}{2 \cdot 4^j} =  \sum_{j \geq k} \frac{1}{2 \cdot 4^j} = \frac{1}{2^{2k+1}} \cdot \sum_{l=0}^\infty \frac{1}{4^l}= \frac{1}{2^{2k+1}}\cdot \frac{4}{3} =  \frac{2^{1-r}}{3} =(-1)^r\cdot \frac{2^{1-r}}{3} \ .
\end{displaymath}
If $r=2k+1$ is odd, then we must also have $m$ odd and hence
\begin{displaymath}
\begin{split}
  \sum_{j \geq k} \frac{1}{2 \cdot 4^j}\cdot (-1)^{m\cdot2^{2j-2k}}
&= -\frac{1}{2 \cdot 4^k}+ \sum_{j > k} \frac{1}{2 \cdot 4^j} =  -\frac{1}{2^{2k+1}}+ \frac{1}{2^{2k+1}}\cdot \frac{4}{3}  -\frac{1}{2^{2k+1}} \\
&= -\frac{2}{3}\cdot\frac{1}{2^{2k+1}}  =(-1)^r \cdot \frac{2^{1-r}}{3} \ .
\end{split}
\end{displaymath}
This shows that for all $x=\frac{m}{2^r}$ with ($m$ odd and $r \geq 1$) or with $r=0$ we have
\begin{displaymath}
\widecheck{1_{W'}}(\psi(x))=(-1)^r\cdot \frac{2^{1-r}}{3} \cdot \exp(-2 \pi i \, x) \,.
\end{displaymath}
Therefore, by Theorem \ref{diff rms} we have with the representation from Eqn.~\eqref{eq:zhalf writing}
\begin{displaymath}
\widehat{\gamma} =\sum_{\frac{m}{2^r} \in \mathbb Z[\frac{1}{2}]} \frac{4^{1-r}}{9} \delta_{\frac{m}{2^r}} \,.
\end{displaymath}
The right panel of Figure \ref{fig:diffperiod} displays the diffraction of a sample, together with labels $(m,r)$.

\section*{Acknowledgements}

NS was supported by NSERC, under grant 2014-03762, and the author is thankful for the support. The manuscript was finalised when NS visited CR at FAU Erlangen--N\"urnberg in summer 2016, and NS would like to thank the Department of Mathematics for hospitality. We thank Michael Baake and Uwe Grimm for useful comments on the manuscript and Tony Guttmann for pointing out reference \cite{NL92} to us. We also thank the referees for useful comments on the manuscript. CR would like to thank the organisers of the conference \textit{Guttmann 2015 -- 70 and Counting} for generous financial support. We thank Conradin Beeli for his kind permission to use the graphics in the left panel of Figure~\ref{fig:diff}.

\appendix
\renewcommand{\theequation}{A\arabic{equation}}
\setcounter{equation}{0}

\section{LCA groups}\label{app:LCA}

We gather basic facts about Fourier analysis on locally compact abelian (LCA) groups. Further background is given  in \cite{RUD, BF, HeRo}, see also \cite[Sec.~2.3.3--2.3.4]{BG2}.

\subsection{LCA groups and their duals}

Cut-and-project schemes $(G,H,\cL)$ may exist in certain general LCA groups. In this article we restrict to Euclidean direct space and to internal space being a compactly generated LCA group, the latter without loss of generality by Remark~\ref{rmsrem}. For any  compactly generated LCA group $G$, the structure theorem \cite[Thm.~9.8]{HeRo} asserts that $G=\mathbb R^d\times\mathbb Z^m\times K$ for some $d,m\in\mathbb N_0$ and some compact group $K$. A relevant example of such $K$ is the product space
\begin{equation}\label{Kref}
K=\prod_{i\in I} \mathbb Z/{n_i\mathbb Z}
\end{equation}
with componentwise addition, where $I$ is some countable index set and $n_i$ are natural numbers. As a countable sum of finite discrete groups, it is indeed compact in the product topology by Tychonoff's theorem.

Any LCA group $G$ admits an invariant measure $\theta_G$, which is unique up to normalisation.  It is called Haar measure. On $\mathbb R^d$, we have the Lebesgue measure, for discrete groups we have the counting measure, and for $K$ in Eq.~\eqref{Kref} we have the product measure of the normalised counting measures on the components.
In order to analyse the frequency of a point set $\Lambda$ in an LCA group $G$, we need a suitable averaging sequence $(A_n)_{n\in\mathbb N}$ generalising balls, i.e., we want to define
\begin{displaymath}
\mathrm{dens}(\Lambda)=\lim_{n\to\infty} \frac{\mathrm{card}(\Lambda \cap A_n)}{\theta_G(A_n)}
\end{displaymath}
if this limit exists. Here $\mathrm{card}(A)$ denotes the number of elements of $A$. For simplicity, in this article we will always take $A_n=\{ (x,y,z) \in \mathbb R^d\times\mathbb Z^m\times K: \max(|x|, |y|) \le n \}$. Our results actually hold for more general so-called van Hove sequences \cite{Martin2}, whose ``boundary-to-bulk ratio'' vanishes in the infinite volume limit.

A character $\chi$ of $G$ is a group homomorphism into the unit circle group $\{z\in\mathbb C: |z|=1\}$. The dual group $\widehat G$ is the set of all \textit{continuous} characters with multiplication as group operation. It is an LCA group when equipped with the topology of compact convergence. For $G=\mathbb R^d$, any continuous character is of the form $\chi_k(x)=e^{2\pi \imath \, k\cdot x}$ for some $k\in \mathbb R^d$. Here $k\cdot x$ denotes the standard scalar product of $k,x\in\mathbb R^d$. We may thus identify $\widehat G$ with $G$ in that case. For $G=\mathbb Z^m$, any character is of the form $\chi_k(x)=e^{2\pi \imath\,  k\cdot x}$ for some $k\in \mathbb T^m=(\mathbb R/\mathbb Z)^m$, the $m$-dimensional torus. We may thus identify $\widehat G$ with $\mathbb T^m$ in that case. For $G=\mathbb Z/n \mathbb Z$, any character is of the form $\chi_k(x)=e^{2\pi \imath \, k\cdot x/n}$ for some $k\in\{0,1,\ldots, n-1\}$. We may thus identify $\widehat G$ with $G$ in that case. Similarly, by requiring continuity of the characters, we may identify the dual of $K$ in Eq.~\eqref{Kref} with the direct sum
\begin{displaymath}
\widehat K\cong\bigoplus_{i\in I} \mathbb Z/{n_i\mathbb Z}\ ,
\end{displaymath}
i.e., any element of $\widehat K$ has only finitely many non-zero components. As a consequence of the above results, we have
\begin{displaymath}
\reallywidehat{ \mathbb R^d\times \mathbb Z^m\times K} \cong \mathbb R^d\times \mathbb T^m\times \widehat K\ .
\end{displaymath}
In fact we have that $\mathbb R^d$, $\mathbb Z^m$, $\mathbb T^m$, $K$ and $\widehat K$ are all isomorphic to their double duals, as this is true of any LCA group.
Also recall that $G$ is discrete if and only if $\widehat{G}$ is compact, and that $G$ is compact if and only if $\widehat{G}$ is discrete. The latter statements all follow from Pontryagin duality \cite{RUD}.

\subsection{Lattices in LCA groups and their duals}

A \textit{lattice} $\mathcal L$ in $G$ is a discrete subgroup of $G$ such that $G/\cL$ is compact. Any lattice admits a relatively compact measurable fundamental domain $F$ such that $G=F+\cL$ is a unique decomposition. A subgroup $H$ of $G$ is a lattice if and only if $H$ is uniformly discrete and relatively dense in $G$, i.e., there is a zero neighborhood $U\subset G$ such that $H\cap U=\{0\}$ and there is  compact $K\subset G$ such that $H+K=G$. Any lattice has a point density $\mathrm{dens}(\cL)$ which equals the reciprocal of the Haar measure of its fundamental domain.

The so-called \textit{dual lattice} $\cL^0\subset \widehat G$  is defined to be the annihilator of $\cL$ in $\widehat G$, i.e.,
\begin{displaymath}
\cL^0 = \{ \chi \in \widehat{G} : \chi(x)= 1 \text{ for all } x \in \cL \}\ .
\end{displaymath}
This is indeed a lattice, as follows from Pontryagin duality $\cL^0\cong \widehat{G/\cL}$ and $ \widehat{G}/\cL^0\cong \widehat{\cL}$, by recalling that compactness and discreteness are dual notions. We have $(\cL^0)^0\cong \cL$, see \cite[Lem.~2.13]{RUD}. For a lattice $\cL$ in $\mathbb R^d$, we thus obtain the dual lattice $\cL^0 = \{ y \in \mathbb R^d : x\cdot y\in \mathbb Z \text{ for all } x \in \cL \}$ by the above identification. 

\section{Positive definite functions}\label{app:pdfctns}

We denote by $C_c(G)$ the space of continuous, compactly supported functions $f:G\to\mathbb C$, and by $C_U(G)$ the space of uniformly continuous and bounded functions $f:G\to\mathbb C$. The expression  $f*g(x)=\int_G f(y)g(x-y){\rm d}\theta_G(y)$ denotes convolution of $f,g\in L^1(G)$, where $L^1(G)$ denotes the space of integrable functions on $G$. We write $\widetilde f(x)=\overline{f(-x)}$ and $f^\dag(x)=f(-x)$. We define the Fourier transform of $f\in L^1(G)$ by  $\widehat f(\chi)=\int_G\overline{\chi}(x) f(x){\rm d}\theta_G(x)$ for $\chi\in\widehat G$. The inverse Fourier transform is denoted by $\widecheck f(\chi)=\int_G \chi(x) f(x){\rm d}\theta_G(x)$.

\begin{definition}[positive definite function] \cite[Def.~3.3]{BF} \label{def:pdfctn}
A function $\varphi:G\to\mathbb C$ is \textit{positive definite} if for any $n\in\mathbb N$, for any $x_1,\ldots, x_n\in G$ and for any $c_1,\ldots, c_n\in \mathbb C$
\begin{displaymath}
\sum_{i,j=1}^n c_i \varphi(x_i-x_j) \overline{c_j} \ge 0.
\end{displaymath}
The space of positive definite continuous compactly supported functions $G\to\mathbb C$ is denoted by $PK(G)$.
\end{definition}

\begin{example}\label{pdfex}
If $f\in L^2(G)$ and compactly supported,  then $f * \widetilde f\in PK(G)$. This holds since
\begin{displaymath}
\sum_{i,j=1}^n c_i f*\widetilde f(x_i-x_j) \overline{c_j} =\int_G \left| \sum_{i=1}^n f(x+x_i)c_i\right|^2 {\rm d}\theta_G(x)\ge 0\ .
\end{displaymath}
\end{example}

Continuous positive definite functions are important in Fourier analysis due to Bochner's theorem\footnote{See Appendix~\ref{app:B} for some background on measures.}.
\begin{theorem}[Bochner]\cite[Thm.~3.12]{BF}, \cite[1.4.3]{RUD}\label{bochthm}
A continuous function $\varphi:G\to\mathbb C$ is positive definite if and only if there exists a positive finite measure $\sigma_\varphi\in\mathcal M(\widehat G)$ such that
\begin{displaymath}
\varphi(x)=\int_{\widehat G} \chi(x) {\rm d} \sigma_\varphi(\chi)\ .
\end{displaymath}
In that case, the measure $\sigma_\varphi$ is uniquely determined.
\qed
\end{theorem}
Using Bochner's theorem, we readily get many examples of continuous positive definite functions from corresponding positive finite measures.
\begin{example}[positive definite functions]\label{expdf2}
\begin{itemize}
\item[(i)] For any character $\chi:G\to \mathbb C$ we have $\sigma_\chi=\delta_\chi$, the Dirac point measure in $\chi\in\widehat G$. In particular, any nonnegative constant function is positive definite.
\item[(ii)] For any pointwise product $\varphi_1\cdot\varphi_2$ where both $\varphi_1$ and $\varphi_2$ are continuous and positive definite, and we have $\sigma_{\varphi_1\cdot\varphi_2}=\sigma_{\varphi_1}*\sigma_{\varphi_2}$.
\item[(iii)] For any convolution $\varphi_1 * \varphi_2$ where $\varphi_1$ is continuous positive definite and $\varphi_2\in PK(G)$, and we have $\sigma_{\varphi_1 * \varphi_2}=\widehat{\varphi_1}\cdot \sigma_{\varphi_2}=\widehat{\varphi_2}\cdot \sigma_{\varphi_1}$, compare \cite[Prop.~4.3]{BF}. Indeed
\begin{displaymath}
\begin{split}
\int_{\widehat{G}} \chi(x) &\,{\rm d} \sigma_{\varphi_1 * \varphi_2} (\chi)= \varphi_1 * \varphi_2 (x) =\int_G \varphi_1(t) \varphi_2 (x-t) \, {\rm d}t \\
&=\int_G \varphi_1(t) \int_{\widehat{G}} \chi(x-t) \, {\rm d} \sigma_{\varphi_2} (\chi) \, {\rm d}t \\
&=\int_{\widehat{G}} \int_G \varphi_1(t)  \overline{\chi(t)} \, {\rm d}t \, \chi(x) \, {\rm d} \sigma_{\varphi_2} (\chi) \\
&=\int_{\widehat{G}}\chi(x) \widehat{\varphi_1}(\chi)  \, {\rm d} \sigma_{\varphi_2} (\chi) =\int_{\widehat{G}}\chi(x)  \, {\rm  d} \left( \widehat{\varphi_1} \cdot \sigma_{\varphi_2} \right) (\chi) \ ,
\end{split}
\end{displaymath}
where Fubini can be used by positivity or because $\varphi_2$ has compact support. This shows that $\sigma_{\varphi_1 * \varphi_2} = \widehat{\varphi_1} \cdot \sigma_{\varphi_2}$. The other identity is shown similarly.

\item[(iv)] For any $f*\widetilde f$ where $f\in L^2(G)$, we have $\sigma_{f*\widetilde f}=|\widecheck f|^2\cdot \theta_G$ by Remark~\ref{pdfex}.
\item[(v)]  For any pointwise product $g\odot h$ where both $g\in C(G)$ and $h\in C(H)$ are positive definite, and we have  $\sigma_{g\odot h}=\sigma_g \otimes \sigma_h$, the product measure of $\sigma_g$ and $\sigma_h$ in $\mathcal M(\widehat G\times \widehat H)$.
\end{itemize}
\end{example}

\section{Positive definite measures}\label{app:B}

\subsection{Measures and linear functionals}

We denote by $\mathcal M(G)$ the space of complex  Radon measures on the compactly generated LCA group $G$. For $\mu$-integrable functions $f: G\to\mathbb C$ we will usually write $\langle \mu, f\rangle$ or $\mu(f)$ instead of  $\int_G f(x) \, {\rm d} \mu(x)$.
 Note that $\mu\in \mathcal M(G)$ gives rise to a linear functional $L_\mu: C_c(G)\to \mathbb C$ via $L_\mu(f)=\mu(f)$. There is an intimate connection between $\cM(G)$ and the space of linear functionals on $C_c(G)$. We will briefly review it, as working with unbounded complex measures does not seem entirely standard. We refer the reader to \cite{Die} for a more detailed discussion on Radon measures and linear functionals on $C_c(G)$.

Recall  that a linear function $L:C_c(G)\to\mathbb C$ is \textit{positive} if $L(f)\ge 0$ for all non-negative $f\in C_c(G)$. For measures, positivity of $L_\mu$ is equivalent to positivity of $\mu$. Now the Riesz Representation Theorem tells us that any positive linear function $L : C_c(G) \to \CC$ is given by a positive regular Borel measure.

 \begin{theorem}\label{Riesz Rep} \cite[Thm.~2.14]{RUD2} Let $L : C_c(G) \to \CC$ be positive and linear. Then, there exists a unique positive regular Borel measure $\mu\in \cM(G)$ such that $L= L_{\mu}$. \qed
\end{theorem}

We give a more general version of this standard result. The key is continuity with respect to the \textit{inductive topology}, i.e., the topology induced by the supremum norm on compact sets. Let us note that $L$ is continuous with respect to the inductive topology if and only if for any compact $K \subseteq G$ there exists a finite constant $C_K$ such that
\begin{equation}\label{eq:induct top}
\left| L(f) \right| \leq C_K \cdot \| f \|_\infty
\end{equation}
for all $f \in C_c(G)$ with $\supp(f) \subseteq K$, where $ \| f \|_\infty=\sup\{|f(x)|:x\in G\}$ denotes the supremum norm of $f$.  We start with two simple observations.

\begin{lemma}\label{variation}
Let $L: C_c(G) \to \CC$ be a linear function which satisfies \eqref{eq:induct top}. For any non-negative $f \in C_c(G)$ we can define
\begin{displaymath}
\left| L \right| (f) := \sup \{ |L(g) | : g \in C_c(G) \text{ such that } |g| \leq f \} <\infty\ .
\end{displaymath}
Then $|L|$ can be extended to a positive linear function on $C_c(G)$. In particular, there exists a unique positive regular Borel measure $\nu\in\cM(G)$ such that $|L|(f)=\int_G f(x) \, {\rm d} \nu(x)$ for all $f \in C_c(G)$.
\end{lemma}
\begin{proof}
Let us first observe that for all non-negative $f \in C_c(G)$ we have $|L|(f) < \infty$ by \eqref{eq:induct top}. Therefore $|L|$ is well defined as a function, and it is straightforward to see that $|L|$ is linear on the cone of non-negative compactly supported functions.
We can then extend $|L|$ to an $\RR$-linear mapping on the space of real valued compactly supported continuous functions via
\begin{displaymath}
|L|(f) = |L|(f_+)-|L|(f_-) \,,
\end{displaymath}
where as usual $f_+(x) =\max \{ f(x) , 0 \}$ and $f_-(x) =-\min \{ f(x) , 0 \}$ are nonnegative functions in $C_c(G)$. Finally we extend $|L|$ to $C_c(G)$ via
\begin{displaymath}
|L|(f) = |L|(\mbox{Re}(f))+i|L|(\mbox{Im}(f)) \,.
\end{displaymath}
\end{proof}
Next we show that positivity implies continuity in the inductive topology.
\begin{lemma}\label{positive implies measure} Let $L : C_c(G) \to \CC$ be positive and linear.  Then $L$ satisfies \eqref{eq:induct top}.
\end{lemma}
\begin{proof} Let $K \subseteq G$ be compact and choose $g \in C_c(G)$ such that $g \geq 1_K$. Then for all $f \in C_c(G)$ with $\supp(f) \subseteq K$ we have
\begin{displaymath}
|\mbox{Re}(f)|, |\mbox{Im}(f)| \leq \|f \|_\infty g \,.
\end{displaymath}
Since $L$ is positive and linear, it preserves inequalities. Therefore
\begin{displaymath}
-L(\|f \|_\infty g) \leq L(\mbox{Re}(g)) , L(\mbox{ Im}(g) ) \leq L( \|f \|_\infty g) \,.
\end{displaymath}
From here, we get that
\begin{displaymath}
\left| L(f) \right|= \left| L(\mbox{Re}(f)+i\mbox{ Im}(f)) \right| \leq 2 L(g)\cdot \| f\|_\infty \,.
\end{displaymath}
We can thus use the constant $C_K= 2 L(g)$, which only depends on $K$.
\end{proof}

Now, we are ready to formulate the general Riesz Representation Theorem \cite{Die}.

\begin{theorem} Let $L: C_c(G) \to \CC$ be any linear function. Then $L$ satisfies \eqref{eq:induct top} if and only if there exists a complex Radon measure $\mu\in \cM(G)$ such that $L= L_{\mu}$.
\end{theorem}
\begin{proof}
``$\Rightarrow$'' This follows immediately from the fact that each complex measure can be written as a linear combination of four positive measures and from Lemma~\ref{positive implies measure}.

``$\Leftarrow$'' If $L$ is real valued, meaning $L(f) \in \RR$ whenever when $f \in C_c(G)$ takes only real values, then by Lemma \ref{variation} we can write $L=|L|-(|L|-L)$ as a linear combination of two positive linear mappings, and the claim follows from Theorem~\ref{Riesz Rep}.
The general case follows now by observing that we can write any linear transformation satisfying \eqref{eq:induct top} as a combination of two real valued linear transformations satisfying \eqref{eq:induct top}.
\end{proof}
Given any measure $\mu\in\cM(G)$, by Lemma \ref{variation} there exists a positive measure $|\mu|\in \cM(G)$, called the \textit{variation} of $\mu$, such that
$|L_{\mu}|= L_{| \mu |}$.
In particular for a point measure $\mu=\sum_{g\in G} c_g\,\delta_g$, its variation is the positive point measure $|\mu|=\sum_{g\in G} |c_g|\,\delta_g$.
A measure $\mu$ is called \textit{finite} if $|\mu|(G)<\infty$.
Given any measure $\mu\in\mathcal M(G)$ and any compact $K\subseteq G$ we define
\begin{displaymath}
\| \mu \|_K := \sup_{x \in G} \left| \mu \right| (x+K) \,.
\end{displaymath}
A measure $\mu$ is called \textit{translation bounded} if for all compact $K \subseteq G$ we have $\| \mu \|_K < \infty$. We denote  by $\cM^\infty(G)$ the space of translation bounded measures on $G$. 

For $\mu,\nu\in \mathcal M(G)$, one of them being bounded, convolution $\mu *\nu\in \mathcal M(G)$ is defined via $\mu * \nu(f)=\int_G\int_G f(x+y) \,{\rm d}\mu(x) \, {\rm d}\nu(y)$. This generalises convolution from functions to measures. Indeed, if
 $f,g$ are integrable functions on $G$, then $(f\cdot\theta_G) * (g\cdot\theta_G)= (f*g)\cdot\theta_G$. 
Similarly, convolution between an integrable function $f$ and a measure $\mu$ is understood 
via $(f*\mu)\cdot\theta_G=(f\cdot\theta_G)*\mu$. Thus $f*\mu$ is a function on $G$.

\subsection{Positive definite measures}
\begin{definition}\cite[Def.~4.2]{BF}
A measure $\mu\in \mathcal M(G)$ is \textit{positive definite} if
\begin{displaymath}
\langle  \mu, f * \widetilde f \rangle \ge 0
\end{displaymath}
for all $f\in C_c(G)$.
\end{definition}
The above definition generalises positive definiteness to measures, compare \cite[Ch.~I]{BF} and \cite[Ch.~4]{ARMA1}.  A continuous function $\varphi:G\to\mathbb C$ is positive definite if and only if the measure $\varphi\cdot \theta_G$ is positive definite \cite[Prop.~4.1]{BF}.

\begin{remark}
If $\mu\in\mathcal M(G)$ is positive definite, then $\langle \mu ,g \rangle\ge0$ for any positive definite $g\in C(G)\cap L^1(\mu)$, compare \cite[Thm.~4.3]{ARMA1}. For $g\in PK(G)$ this follows from Lemma~\ref{pdpk}.
\end{remark}
The following examples are direct consequences of the definition.
\begin{example}[positive definite measures]
\begin{itemize}

\item[(i)] The measure $\delta_0$ is positive definite as
\begin{displaymath}
\langle \delta_0, f * \widetilde f\rangle = f*\widetilde f (0)=\int_G |f(x)|^2 {\rm d}x\ge 0\ .
\end{displaymath}
\item[(ii)] Any Haar measure $\theta_G$ is positive definite as $\langle \theta_G, f * \widetilde f\rangle = |\langle \theta_G, f \rangle|^2 \ge 0$.
\item[(iii)] Let $H\subseteq G$ be a subgroup of $G$ and let $\theta_H$ be a Haar measure on $H$. Then $\theta_H$ is positive definite on $H$. But $\theta_H$, viewed as a measure on $G$, is also positive definite. This holds since for $f\in C_c(G)$, we have $f|_H\in C_c(H)$ for its restriction to $H$.
\item[(iv)] If $\mu$ is a finite measure, then $\mu * \widetilde \mu$ is positive definite as
\begin{displaymath}
\langle \mu * \widetilde \mu , f * \widetilde f \rangle = \langle \theta_G, |\widetilde f * \mu|^2 \rangle \ge 0 \ .
\end{displaymath}
\end{itemize}
\end{example}

\begin{lemma}\label{pdpk} Let $\mu \in \mathcal M(G)$ be positive definite and $g \in PK(G)$. Then $\mu*g$ is continuous and positive definite.
\end{lemma}

\begin{proof}
Take any $f\in C_c(G)$ and note that $\mu * f* \widetilde f$ is positive definite as $\langle\mu * f* \widetilde f, g*\widetilde g \rangle=\langle \mu, (f^\dag*g)*\widetilde{(f^\dag*g)}\rangle\ge 0$ for any $g\in C_c(G)$, compare \cite[Prop.~4.4]{BF}. Now apply Remark~\ref{expdf2} (iii) to $\varphi_1=\mu * f* \widetilde f$ and $\varphi_2=g$. If follows that $(\mu*g)*f*\widetilde f=\varphi_1*\varphi_2$ is positive definite. In particular we have $\langle \mu*g, f*\widetilde f\rangle=(\mu*g)*f*\widetilde f(0)\ge 0$.
\end{proof}

\begin{corollary} Let $\mu \in \mathcal M(G)$ be positive definite and $\nu$ be a positive definite measure with compact support. Then $\mu*\nu$ is positive definite.
\end{corollary}
\begin{proof}
Let $f \in C_c(G)$. Then $\nu*f*\widetilde{f}$ is positive definite \cite{BF} and hence $\nu*f*\widetilde{f} \in PK(G)$. Then, by Lemma \ref{pdpk} we get that $\mu*(\nu*f *\widetilde{f})$ is a positive definite function. This shows that
\[
\langle \mu*\nu , f*\widetilde f\rangle=(\mu*\nu)*f*\widetilde f(0)\ge 0 \,.
\]
\end{proof}
The following theorem is the measure analogue of Bochner's theorem (Theorem~\ref{bochthm}) for positive definite functions.
\begin{theorem}\label{Bochmeas}\cite[Thm.~4.7]{BF}
A measure $\mu\in \mathcal M(G)$ is positive definite if and only if there exists a positive measure $\sigma_\mu\in \mathcal M(\widehat G)$ such that
\begin{equation}\label{form:trans}
\langle \mu, f*\widetilde f\rangle = \langle \sigma_\mu, |\widecheck f|^2\rangle
\end{equation}
for all $f\in C_c(G)$. In that case, the measure $\sigma_\mu$ is uniquely determined. It is called the Fourier transform of $\mu$. We also write $\sigma_\mu=\widehat\mu$.
\qed
\end{theorem}

\begin{remark}\label{rem:FT}
\begin{itemize}
\item[(i)] For the convenience of the reader, we will not use the above theorem in this article but develop Fourier theory for model sets from scratch, based on Bochner's theorem for positive definite functions. The first step of our proof of the Poisson summation formula in Proposition~\ref{psflattice}, however, is a crucial ingredient of a proof of Theorem~\ref{Bochmeas}.

\item[(ii)]
There exists a theory of positive definite measures on certain non-abelian groups, which is also based on some version of Bochner's theorem \cite[Ch.~9]{Wolf}. It has recently been
used for harmonic analysis of non-abelian model sets \cite{BHP2}.

\item[(iii)]
Fourier analysis of unbounded measures beyond the positive definite case is developed by Argabright and de Lamadrid \cite{ARMA1}. There transformability of $\mu\in\mathcal M(G)$ is defined by requiring Eq.~\eqref{form:trans} to hold for some $\sigma_\mu\in \mathcal M(\widehat G)$. The transform is always a translation bounded measure \cite[Thm.~2.5]{ARMA1}.   It is an open question whether any such $\mu$ is a linear combination of positive definite measures, compare the discussion in \cite[p.~29]{ARMA1}.

\item[(iv)]
For example, Eq.~\eqref{form:trans} with $\mu=\delta_g$ reduces to the Fourier inversion formula, and we get $\widehat\mu=g\cdot\theta_{\widehat G}$, where $g:\widehat G\to\mathbb C$ is understood as $g(\chi)=\chi(g)$.

\end{itemize}
\end{remark}


\begin{thebibliography}{99}

\bibitem{ARMA1} L.N. ~Argabright and J.~Gil ~de ~Lamadrid, \textit{Fourier Analysis of Unbounded Measures on Locally Compact Abelian Groups}, Memoirs of
the Amer. Math. Soc. \textbf{145}, Amer. Math. Soc., Providence, RI, 1974.

\bibitem{BG1}
\newblock M.~Baake and~U.~Grimm,
\textit{Kinematic diffraction from a mathematical viewpoint},
Z.~Kristallogr.~\textbf{226}, 711--725, 2011. \texttt{arXiv:1105.0095}

\bibitem{BG3}
\newblock M.~Baake and~U.~Grimm,
\textit{Mathematical diffraction of aperiodic structures},
Chem.~Soc.~Rev.~\textbf{41}, 6821--6843, 2012. \texttt{arXiv:1205.3633}

\bibitem{BG2}
\newblock M.~Baake and U.~Grimm,
\newblock \textit{Aperiodic Order. Vol. 1. A Mathematical Invitation},
\newblock Encyclopedia of Mathematics and its Applications \textbf{149},
\newblock Cambridge University Press, Cambridge, 2013.

\bibitem{BHP}
M.~Baake, J.~Hermisson and P.A.B.~Pleasants, \textit{The torus parametrization of quasiperiodic LI-classes}, J.~Phys.~A: Math.~Gen.~\textbf{30}, 3029--3056, 1997.

\bibitem{BH}
M.~Baake and C.~Huck, \textit{Ergodic properties of visible lattice points},
Proceedings of the Steklov Institute of Mathematics \textbf{288}, 165--188, 2015.
\texttt{arXiv:1501.01198}


\bibitem{BHS}
M.~Baake, C.~Huck and N.~Strungaru, \textit{On weak model sets of extremal density}, to appear in Indag.~Math., 2016.
\texttt{arXiv:1512.07129v2}

\bibitem{BJL}
M.~Baake, T.~J\"ager and D.~Lenz, \textit{Toeplitz flows and model sets}, published online in Bull. London Math. Soc., 2016.
\texttt{arXiv:1511.08595v1}

\bibitem{BM} M. Baake and R.V.~Moody,  \textit{Weighted Dirac combs with pure
point diffraction}, J. Reine Angew. Math. ~\textbf{573}, 61--94, 2004.
\texttt{arXiv:math/0203030}

\bibitem{bmp00}
\newblock M.~Baake, R.V.~Moody and P.A.B.~Pleasants,
\newblock \textit{Diffraction from visible lattice points and kth power free integers},
\newblock Discrete Math. \textbf{221}, 3--42, 2000.

\bibitem{BMRS}
\newblock M.~Baake, R.V.~Moody, C.~Richard and B.~Sing, \textit{Which distributions of matter diffract? -- Some answers}, in \textit{Quasicrystals: Structure and Physical Properties}, H.-R. Trebin (ed.), Wiley-VCH, Berlin, 188--207, 2003. \texttt{arXiv:math-ph/0301019}.

\bibitem{BMS}
\newblock M.~Baake, R.V.~Moody and M.~Schlottmann,
\newblock \textit{Limit-(quasi)periodic point sets as quasicrystals with p-adic internal spaces},
\newblock J. Phys. A \textbf{31}, 5755--5765, 1998.
\texttt{arXiv:math-ph/9901008}

\bibitem{BKKL2015}
A.~Bartnicka, S.~Kasjan, J.~Ku\l{}aga-Przymus, and M.~Lema\'n{}czyk.
\newblock \textit{$B$-free sets and dynamics}, preprint, 2015.
\texttt{arXiv:1509.08010}

\bibitem{BF}
C. ~Berg and G. ~Forst, \textit{Potential Theory on Locally Compact
Abelian Groups}, Springer, Berlin, 1975.

\bibitem{Bi09}
L.~Bindi, P.J.~Steinhardt, N.~Yao and P.J.~Lu, \textit{Natural quasicrystals}, Science \textbf{324}, 1306--1309, 2009.

\bibitem{Bi15}
L.~Bindi et al., \textit{Natural quasicrystal with decagonal symmetry}, Scientific Reports \textbf{5}, 9111, 2015.

\bibitem{BHP2}
M.~Bj\"orklund, T.~Hartnick, and F.~Pogorzelski.
\newblock \textit{Aperiodic order and spherical diffraction}, preprint, 2016.
\texttt{arXiv:1602.08928}

\bibitem{dBr}
N.G.~de~Bruijn, \textit{Algebraic theory of Penrose's non-periodic tilings of the plane. I, II}, Indag.~Math.~\textbf{84}, 39--66, 1981.

\bibitem{C}
J.M.~Cowley, \textit{Diffraction Physics}, North-Holland, Amsterdam, second edition, 1990.

\bibitem{Die}
J.~Dieudonn\'e, \textit{Treatise on analysis}, vol. 2,  Academic Press, New York, 1970.


\bibitem{DK}
M.~Duneau and A.~Katz, \textit{Quasiperiodic patterns},  Phys.~Rev.~Lett. \textbf{54}, 2688--1691, 1985.

\bibitem{E}
V.~Elser, \textit{Comment on ``Quasicrystals: A new class of ordered structures''}, Phys.~Rev.~Lett.~\textbf{54}, 1730, 1985.

\bibitem{GR1}
F.~G\"ahler and J.~Rhyner,  \textit{Comment on ``Structure of rapidly quenched Al-Mn''}, Phys.~Rev.~Lett.~\textbf{55}, 2369, 1985.

\bibitem{GR}
F.~G\"ahler and J.~Rhyner,  \textit{Equivalence of the generalised grid and projection methods for the construction of quasiperiodic tilings}, J.~Phys.~A:~Math.~Gen.~\textbf{19}, 267--277, 1986.

\bibitem{GW}
A.J.~Guttmann and H.~Wagenfeld, \textit{A theoretical calculation of x-ray absorption cross sections}, Acta Cryst.~\textbf{22}, 334, 1967.

\bibitem{Hen}
C.L.~Henley, \textit{Random tiling models}, in: Quasicrystals: The State of the Art, eds. D.P.~DiVincenzo and P.~Steinhardt, World Scientific, Singapore, 429--524, 1991.

\bibitem{HeRo}
E.~Hewitt  and K.A.~Ross,
\textit{Abstract Harmonic Analysis. Vol. I.}, Springer, Berlin, 1979.

\bibitem{Hof1}
A.~Hof, \textit{On diffraction by aperiodic structures}, Comm.~Math.~Phys. \textbf{169}, 25--43, 1995.

\bibitem{HR15}
C.~Huck and C.~Richard, \textit{On pattern entropy of weak model sets}, Discrete~Comput.~Geom. \textbf{54}, 741--757, 2015.
\texttt{arXiv:1412.6307}

\bibitem{I}
T.~Ishimasa, H.~Nissen and Y.~Fukano, \textit{New ordered state betwenn crystalline and amorphous in Ni-Cr particles}, Phys.~Rev.~Lett.~\textbf{55}, 511--513, 1985.

\bibitem{KKL} P.A.~Kalugin, A.Yu.~Kitaev and L.S.~Levitov, \textit{$Al_{0.86}Mn_{0.14}$: a six-dimensional crystal}, Zh.~Eksp.~Theor.~Fiz. Red. \textbf{41}, 119--121, JETP Lett. \textbf{41}, 145--149, 1985.

\bibitem{KR}
G.~Keller and C.~Richard, \textit{Dynamics on the graph of the torus parametrisation}, to appear in Ergodic Theory and Dynam. Systems, 2016. \texttt{arXiv:1511.06137}

\bibitem{K1}
P.~Kramer, \textit{Non-periodic central space filling with icosahedral symmetry using copies of seven elementary cells}, Acta~Cryst.~A\textbf{38}, 257--264, 1982.

\bibitem{KN}
P.~Kramer and R.~Neri, \textit{On periodic and non-periodic space fillings of $\mathbb E^m$ obtained by projection}, Acta~Cryst.~A\textbf{40}, 580--587, 1984, and \textit{Erratum}, Acta~Cryst.~A\textbf{41}, 619, 1985.

\bibitem{LP}
M.~Ladd and R.~Palmer, \textit{Structure Determination by X-ray Crystallography}, Springer, New York, 2013.

\bibitem{L}
J.C.~Lagarias, \textit{Meyer's concept of quasicrystal and quasiregular sets}, Comm.~Math.~Phys. \textbf{179}, 365--376, 1996.

\bibitem{ARMA}
J. G. de ~Lamadrid and L.N. ~Argabright, \textit{Almost Periodic
Measures}, Memoirs of the Amer. Math. Soc., \textbf{85}, No. 428, 1990.

\bibitem{LR}
D. Lenz and C. Richard, \textit{Pure point diffraction and cut-and-project schemes for measures: The smooth case}, Mathematische Zeitschrift \textbf{256},  347--378, 2007.
\texttt{arXiv:math/0603453}

\bibitem{LS2}
D. Lenz  and N. Strungaru, \textit{On Weakly Almost Periodic Measures}, preprint, 2016.
\texttt{arXiv:1609.08219}

\bibitem{LO}
\newblock N.~Lev and A.~Olevskii,
\newblock \textit{Quasicrystals and Poisson's summation formula},
\newblock Invent. Math. \textbf{200}, 585--606, 2015.
\texttt{arXiv:1312.6884}


\bibitem{LS84}
D.~Levine and P.~Steinhardt, \textit{Quasicrystals: a new class of ordered structures}, Phys.~Rev.~Lett.~\textbf{53}, 2477--2479, 1984.

\bibitem{LD07}
R.~Lifshitz and H.~Diamant, \textit{Soft quasicrystals -- Why are they stable?}, Phil.~Mag.~\textbf{87}, 3021--3030, 2007.

\bibitem{me72}
Y.~Meyer, \textit{Algebraic Numbers and Harmonic Analysis}, North-Holland, Amsterdam, 1972.

\bibitem{me73}
Y.~Meyer, \textit{Adeles et series trigonometriques speciales}, Annals Math. \textbf{97},
171--186, 1973.

\bibitem{RVM3}
R.V.~Moody, \textit{Meyer sets and their duals}, In: \textit{The Mathematics of Long-Range
Aperiodic Order}, (R. V. Moody, ed.), NATO ASI Series \textbf{C 489},
Kluwer, Dordrecht, 403--441, 1997.

\bibitem{MoSt}
R.V.~Moody and N.~Strungaru, \textit{Almost Periodic Measures and their Fourier
Transforms}, to appear in:  \textit{Aperiodic Order. Vol. 2. Crystallography and Almost Periodicity}, eds. M.~Baake and U.~Grimm, Cambridge University Press, Cambridge, 2016.

\bibitem{NL92}
B.W.~Ninham and S.~Lidin, \textit{Some remarks on quasi-crystal strucure}, Acta~Cryst.~A\textbf{48}, 640--649, 1992.

\bibitem{P89}
G.K.~Pedersen, \textit{Analysis Now}, Graduate Texts in Mathematics \textbf{118}, Springer, New York, 1989.

\bibitem{P}
R.~Penrose, \textit{The r\^ole of aesthetics in pure and applied mathematical research},  Bull.~Inst.~Math.~Appl. \textbf{10}, 266--271, 1974.

\bibitem{PH}
P.A.B.~Pleasants and C.~Huck, \textit{Entropy and diffraction
of the $k$-free points in $n$-dimensional lattices}, Discrete
Comput. Geom. \textbf{50}, 39--68 (2013). \texttt{arXiv:1112.1629}

\bibitem{Rei2}
H. Reiter and  J.D. Stegeman, {\em Classical Harmonic Analysis and Locally Compact Groups}, Clarendon Press, Oxford, 2000.

\bibitem{CR99}
C.~Richard,  \textit{An alternative view on quasicrystalline random tilings}, J.~Phys. A:~Math. Gen. \textbf{32},  8823--8829, 1999.
 \texttt{arXiv:cond-mat/9907262}

\bibitem{CR}
C.~Richard, \textit{Dense Dirac combs in Euclidean space
with pure point diffraction}, J. Math. Phys. \textbf{44}, 4436--4449, 2003.
\texttt{arXiv:math-ph/0302049}

\bibitem{RHHB}
C.~Richard, M.~H\"offe, J.~Hermisson and M.~Baake, \textit{Random tilings: concepts and examples}, J.~Phys.~A:~Math.~Gen. \textbf{31}, 6385--6408, 1998.
\texttt{arXiv:cond-mat/9712267}

\bibitem{RS2}
C.~Richard and N.~Strungaru, \textit{Pure Point Diffraction and Poisson Summation}, preprint, 2015.
\texttt{arXiv:1512.00912}

\bibitem{RUD}
W.~Rudin, \textit{Fourier Analysis on Groups}, Interscience Publishers, New York, 1962.

\bibitem{RUD2}
W.~Rudin, \textit{Real and Complex Analysis}, McGraw-Hill, New York, 1987.


\bibitem{Martin2}
M.~Schlottmann, \textit{Generalized model sets and dynamical
systems}, in: \textit{Directions in Mathematical Quasicrystals}, eds. M.~Baake and R.V.~Moody, CRM Monogr. Ser., Amer. Math. Soc., Providence, RI, 143--159, 2000.

\bibitem{Se}
M.~Senechal, \textit{The mysterious Mr.~Ammann}, Math.~Intelligencer \textbf{26}, 10--21, 2004.

\bibitem{S}
D.~Shechtman, I.~Blech, D.~Gratias and J.W.~Cahn, \textit{Metallic phas with long-range orientational order and no translational symmetry}, Phys.~Rev.~Lett.~\textbf{53}, 1951--1953, 1984.

\bibitem{SB}
D.~Shechtman and I.~Blech, \textit{The microstructure of rapidly solidified $Al_6 Mn$}, Met.~Trans.~A \textbf{16A}, 1005--1012, 1985.

\bibitem{sing07}
\newblock B.~Sing,
\newblock  \textit{Pisot Substitutions and Beyond},
\newblock Dissertation, University of Bielefeld  (2007), \texttt{http://pub.uni-bielefeld.de/publication/2302336}.

\bibitem{NS11}
N.~Strungaru, \textit{ Almost periodic measures and Meyer sets}, to appear in:  \textit{Aperiodic Order. Vol. 2. Crystallography and Almost Periodicity}, eds. M.~Baake and U.~Grimm, Cambridge University Press, Cambridge, 2016.
\texttt{arXiv:1501.00945}.


\bibitem{Wolf}
J.~A. Wolf, \textit{Harmonic analysis on commutative spaces}, volume 142 of \textit{ Mathematical Surveys and Monographs}.
\newblock American Mathematical Society, Providence, RI, 2007.

\end{thebibliography}
\end{document}